\documentclass[12pt]{amsart}

\usepackage{amsmath,amsthm,amsfonts,amssymb,eucal, bbm}
\usepackage{scalerel}

\usepackage{color}

\renewcommand {\a}{ \alpha }
\renewcommand{\b}{\beta}

\newcommand{\g}{\gamma}
\newcommand{\G}{\Gamma}

\renewcommand{\d}{\delta}
\newcommand{\s}{\sigma}
\renewcommand{\l}{\lambda}

\newcommand{\z}{\zeta}
\renewcommand{\t}{\theta}

\newcommand{\p}{\partial}
\newcommand{\om}{\omega}
\newcommand{\Om}{\Omega}

\newcommand{\oq}{\ {\raise 7pt\hbox{${\scriptstyle\circ}$}}
	\kern -7pt{
		\hbox{$Q$}}}

\newcommand{\R}{ \mathbb R}

\newcommand {\BS}{\mathbf S}

\newcommand {\bx}{\mathbf x}

\newcommand {\by}{\mathbf y}

\newcommand{\SfG}{{\sf{G}}}
\newcommand{\sg}{{\sf{g}}}

\newcommand{\SU}{{\sf{U}}}

\newcommand{\scalel}[1]
{{\scaleto{#1}{3pt}}}
\newcommand{\scalet}[1]
{{\scaleto{#1}{4pt}}}

\newcommand{\SfS}{{\sf S}}

\newcommand {\BOLG}{\boldsymbol\Gamma}

\newcommand {\BPsi}{\boldsymbol\Psi}

\newcommand{\CK}{\mathcal K}
\newcommand{\CB}{\mathcal B}

\newcommand{\CT}{\mathcal T}

\newcommand{\CH}{\mathcal H}

\newcommand{\CP}{\mathcal P}

\newcommand{\CA}{\mathcal A}

\newcommand{\CM}{\mathcal M}

\newcommand{\CC}{\mathcal C}

\newcommand{\plainW}[1]{\textup{{\textsf{W}}}^{#1}}
\newcommand{\plainC}[1]{\textup{{\textsf{C}}}^{#1}}

\newcommand{\plainH}[1]{\textup{{\textsf{H}}}^{#1}}

\newcommand{\plainL}[1]{\textup{{\textsf{L}}}^{#1}}

\DeclareMathOperator{\iop}{{\sf Int}}

\DeclareMathOperator{\supp}{{supp}}

\DeclareMathOperator{\dc}{d}

\newtheorem{thm}{Theorem}[section]
\newtheorem{cor}[thm]{Corollary}
\newtheorem{lem}[thm]{Lemma}
\newtheorem{prop}[thm]{Proposition}

\theoremstyle{definition}
\newtheorem{example}[thm]{Example}

\newtheorem*{remark}{Remark}
\newtheorem{rem}[thm]{Remark}

\numberwithin{equation}{section}

\newcommand{\bee}{\begin{equation}}
	\newcommand{\ene}{\end{equation}}
\newcommand{\bees}{\begin{equation*}}
	\newcommand{\enes}{\end{equation*}}
\newcommand{\bes}{\begin{split}}
	\newcommand{\ens}{\end{split}}

\newcommand{\bet}{\begin{thm}}
	\newcommand{\ent}{\end{thm}}
\newcommand{\bel}{\begin{lem}}
	\newcommand{\enl}{\end{lem}}
\newcommand{\bec}{\begin{cor}}
	\newcommand{\enc}{\end{cor}}
\newcommand{\bep}{\begin{proof}}
	\newcommand{\enp}{\end{proof}}
\newcommand{\ber}{\begin{rem}}
	\newcommand{\enr}{\end{rem}}

\newcommand{\CF}{\mathcal F}

\begin{document}
	\hoffset -4pc

\title
[ Eigenvalue asymptotics]
{{Eigenvalue asymptotics for the one-particle density matrix}} 
\author{Alexander V. Sobolev}
\address{Department of Mathematics\\ University College London\\
	Gower Street\\ London\\ WC1E 6BT UK}
 \email{a.sobolev@ucl.ac.uk}



\date{\today}
\begin{abstract}	 
The one-particle density matrix $\gamma(x, y)$ for a bound state of an atom or molecule 
is one of the key objects in the quantum-mechanical approximation schemes. 
We prove the asymptotic formula $\l_k \sim (Ak)^{-8/3}$, $A \ge 0$, 
as $k\to\infty$,   
for the eigenvalues $\l_k$ of the self-adjoint operator $\BOLG\ge 0$ with kernel $\g(x, y)$. 
\end{abstract}

\maketitle

\section{Introduction}

Consider on $\plainL2(\R^{3N})$ the Schr\"odinger operator 
\begin{align}\label{eq:ham}
\CH = \sum_{k=1}^N \bigg(-\Delta_k -  \frac{Z}{|x_k|}
\bigg) 
 + \sum_{1\le j< k\le N} \frac{1}{|x_j-x_k|},
\end{align}
describing an atom with $N$ particles 
(e.g. electrons)  
with coordinates $\bx = (x_1, x_2, \dots, x_N)$, $x_k\in\R^3$, $k= 1, 2, \dots, N$, 
and a nucleus with charge $Z>0$. The notation $\Delta_k$ is used for 
the Laplacian w.r.t. the variable $x_k$. 
The operator $\CH$ acts on the Hilbert space $\plainL2(\R^{3N})$ and it is self-adjoint on the domain 
$D(\CH) =\plainH2(\R^{3N})$, since the potential in \eqref{eq:ham} 
is an infinitesimal perturbation 
relative to the unperturbed operator $-\Delta = - \sum_k \Delta_k$, 
see e.g. \cite[Theorem X.16]{ReedSimon2}.  
Let $\psi = \psi(\bx)$,  
be an eigenfunction of the operator $\CH$ with an eigenvalue $E\in\R$, i.e. $\psi\in D(\CH)$ and 
\begin{align*}
(\CH-E)\psi = 0.
\end{align*}
For each $j=1, \dots, N$, we represent
\begin{align*}
\bx = (\hat\bx_j, x_j), \quad \textup{where}\ 
\hat\bx_j = (x_1, \dots, x_{j-1}, x_{j+1},\dots, x_N),
\end{align*}
with obvious modifications if $j=1$ or $j=N$. 
The one-particle density matrix is defined as the function 
\begin{align}\label{eq:den}
\g(x, y) = \sum_{j=1}^N\int\limits_{\R^{3N-3}} \overline{\psi(\hat\bx_j, x)} \psi(\hat\bx_j, y)\  
d\hat\bx_j,\quad (x,y)\in\R^3\times\R^3. 
\end{align} 
This function is one of the key objects in 
the  multi-particle quantum mechanics, see 
\cite{RDM2000}, \cite{Davidson1976}, 
\cite{LLS2019}, \cite{LiebSei2010} for details and futher references. 
If one assumes that all $N$ particles are spinless fermions (resp. bosons), i.e. that the function $\psi$ is 
antisymmetric (resp. symmetric) under the permutations $x_j\leftrightarrow x_k$, 
then the definition \eqref{eq:den} simplifies:
\begin{align}\label{eq:fb}
\g(x, y) = N \int_{\R^{3N-3}} \overline{\psi(\hat\bx, x)} \psi(\hat\bx, y) d\hat\bx,\ 
\quad \textup{where} \ \hat\bx = \hat\bx_N. 
\end{align}
Our main result however does not require any symmetry assumptions. 
For the sake of completeness mention that, as found in 
\cite{HearnSob2020}, the function \eqref{eq:den} 
is real-analytic for all $x\not = 0, y\not = 0, x\not = y$. 
In the current paper our focus is on spectral properties of the self-adjoint non-negative 
operator $\BOLG$ with the kernel $\g(x, y)$, which we 
call \textit{the one-particle density operator}. 
The operator $\BOLG$ is easily shown to be trace class, and in \cite{Fries2003} 
it was shown that $\BOLG$ has infinite rank. 
However no sharp results on the 
behaviour of the eigenvalues $\l_k(\BOLG)>0$ as $k\to\infty$ had been available until paper 
\cite{Sobolev2020} (see however 
\cite{Cioslowski2020}, \cite{CioPrat2019} for 
relevant quantum chemistry calculations), where it was shown that $\l_k(\BOLG) = O(k^{-8/3})$. 
We always label eigenvalues in non-increasing order counting multiplicity. 
The purpose of the paper is to prove the asymptotic formula \eqref{eq:main}, which 
confirms the sharpness of the bound from \cite{Sobolev2020}. 
Apart from being a mathematically interesting and challenging question, spectral asymptotics 
for the operator $\BOLG$ are important 
%
%
for electronic structure computations as it limits accuracy of 
electronic properties computed with finite basis sets, see e.g. 
\cite{Cioslowski2020}, 
\cite{CioStras2021},  \cite{Fries2003} and \cite{HaKlKoTe2012}
for discussion. 

We assume throughout that  $\psi$ decays exponentially as $|\bx|\to \infty$:
\begin{align}\label{eq:exp}
|\psi(\bx)|\lesssim e^{-\varkappa_\scalel{0} |\bx| },\ \bx\in\R^{3N}.
\end{align}  
Here $\varkappa_0 >0$ is a constant, and the notation ``$\lesssim$" means that the left-hand side 
is bounded from above by the right-hand side times 
some positive constant whose precise value is of 
no importance for us. This notation is used throughout the paper. The property 
\eqref{eq:exp} holds for the eigenfunctions associated with discrete 
eigenvalues (i.e. the ones below the essential spectrum), and in particular, for the ground state. 
For references and detailed discussion we quote \cite{SimonSelecta}.
 
The next theorem contains a concise version of the main result. 

\begin{thm}\label{thm:main} 
Suppose that the eigenfunction $\psi$ satisfies the bound \eqref{eq:exp}. 
Then the eigenvalues $\l_k(\BOLG),k = 1, 2, \dots$, of the operator $\BOLG$ with kernel 
\eqref{eq:den} satisfy the relation
\begin{align}\label{eq:main}
\lim_{k\to\infty}k^{\frac{8}{3}}\l_k(\BOLG) = A^{\frac{8}{3}},
\end{align}
with an explicit constant $A\ge 0$. 
\end{thm}

The complete statement includes a formula for the coefficient $A$, and it is given 
as Theorem \ref{thm:maincompl}. 
 
\begin{remark}
Theorem \ref{thm:main} 
extends to the case of a molecule with several nuclei whose positions
are fixed, i.e. the operator \eqref{eq:ham} can be replaced by 
\begin{align*}
\CH = \sum_{k=1}^N \bigg(-\Delta_k -  \sum_{l=1}^{N_0}\frac{Z_l}{|x_k-R_l|}
\bigg) 
 + \sum_{1\le j< k\le N} \frac{1}{|x_j-x_k|},
\end{align*}
with constant $R_l\in\R^3$ and nuclear charges $Z_l>0$, $l = 1, 2, \dots, N_0$. 
The modifications are straightforward.
\end{remark}

Let us outline the main ideas of the proof. First we represent the operator 
$\BOLG$ as the product $\BOLG = \BPsi^*\BPsi$, where the operator
$\BPsi:\plainL2(\R^3)\to \plainL2(\R^{3N-3})$ with a vector-valued kernel is defined 
in Subsect. \ref{subsect:fact}. 
Therefore we have $\l_k(\BOLG) = s_k(\BPsi)^2, k = 1, 2, \dots$, where 
$s_k(\BPsi)$ are the singular values ($s$-values) of the operator $\BPsi$. As a consequence, 
the asymptotic formula \eqref{eq:main} rewrites as 
\begin{align}\label{eq:main1}
\lim_{k\to\infty}k^{\frac{4}{3}} s_k(\BPsi) = A^{\frac{4}{3}}.
\end{align}
For the sake of discussion 
consider the fermionic (or bosonic) case, in which the 
kernel $\g(x, y)$ is given by \eqref{eq:fb}. Then it is 
straightforward that $\BOLG = \BPsi^*\BPsi$ with the operator 
\begin{align}\label{eq:ferbos}
(\BPsi u)(\hat\bx) = \sqrt{N}\int_{\R^3} \psi(\hat\bx, x) u(x) d x,\  u\in\plainL2(\R^3).
\end{align}
For integral operators 
the rate of decay of singular values increases with the smoothness of their 
kernels, and the appropriate estimates  
via suitable Sobolev norms can be found in \cite{BS1977}. 
Such estimates, 
together with the recent regularity estimates for $\psi$ obtained in \cite{FS2018},  
were used in \cite{Sobolev2020} to prove the bound 
$s_k(\BPsi)\lesssim k^{-4/3}$, $k = 1, 2, \dots$. 

The study of spectral asymptotics of the operator \eqref{eq:ferbos} requires 
more precise information on the singularities of $\psi$. 
By elliptic regularity, the function $\psi$ is 
real analytic away from the coalescence points of 
the particles, i.e. for $x_j\not = x_k, 1\le j < k \le N$ and $x_j\not = 0$, 
$j = 1, 2, \dots, N$, and hence only the coalescence points contribute to the asymptotics \eqref{eq:main1}.  
As shown by T. Kato in \cite{Kato1957}, the function $\psi$ is Lipschitz. 
Of course, this fact alone is not sufficient to obtain an 
asymptotic formula for $\BPsi$ -- one needs to know 
the precise shape of the function $\psi$ near the coalescence points. 
A suitable representation formula for the function $\psi$ was obtained 
in \cite{FHOS2009}. To explain in more detail we make a further simplifying assumption and 
consider the special case $N=2$, so that $\bx = (t, x)\in \R^3\times\R^3$, 
and the operator $\BPsi$ acts from $\plainL2(\R^3)$ into $\plainL2(\R^3)$. 
According to \cite{FHOS2009}, there exists a neighbourhood (open connected set) 
$\Om_{1,2}\subset \big(\R^3\setminus \{0\}\big)\times \big(\R^3\setminus\{0\}\big)$ of the 
diagonal set $\{(x, x): x\in\R^3\setminus \{0\}\}$ 
and two functions $\xi_{1,2}, \eta_{1, 2}$, real analytic in $\Om_{1,2}$,  such that 
the eigenfunction $\psi = \psi(t, x)$ admits the representation
\begin{align}\label{eq:locan}
\psi(t, x) = \xi_{1,2}(t, x) + |t-x|\,\eta_{1,2}(t, x),
\quad \textup{for all}\quad (t, x)\in \Om_{1,2}.
\end{align}
The form of the second term is in line with Kato's observation 
(see \cite{Kato1957}) that $\psi$ is Lipschitz. 
The representation \eqref{eq:locan} is 
ideally suited for the study of spectral asymptotics. 
Indeed, 
the Lipschitz factor on the right-hand side of \eqref{eq:locan} is homogeneous of order one.  
The behaviour of eigenvalues for a wide class of integral operators 
including those with homogeneous kernels, 
was studied by 
M. Birman and M. Solomyak in \cite{BS1970},\cite{BS1977_1} and \cite{BS1979}, 
see also \cite{BS1977}. However, the existing results are not directly applicable, since 
the functions 
$\xi_{1, 2}$ and $\eta_{1, 2}$ may not be smooth on the closure 
$\overline{\Om_{1,2}}$. 
Moreover, there is no information on the 
integrability of $\xi_{1, 2}$ and $\eta_{1, 2}$ over $\Om_{1, 2}$.  
 To circumvent this difficulty we 
approximate $\xi_{1,2}, \eta_{1,2}$ by suitable 
$\plainC\infty_0$-functions supported inside $\Om_{1,2}$. 
The error incurred is controlled with the help of the bounds 
obtained in \cite{Sobolev2020}. Using 
the Birman-Solomyak results 
and subsequently 
taking the limit of these smooth approximations we arrive at the formula 
\eqref{eq:main1} with the coefficient
\begin{align*}
A = \frac{1}{3}\bigg(\frac{2}{\pi}\bigg)^{\frac{5}{4}}\int_{\R^3} |2^{1/2}\eta_{1,2}(x, x)|^{3/4} dx.
\end{align*}  
The finiteness of the above integral is a by-product of the proof. Note that the coalescence points  
$x = 0$ and $t=0$ do not affect the asymptotics. 
 
For $N\ge 3$ application of the existing results on spectral asymptotics for integral operators 
is not immediate. It relies on the reduction to a 
certain model operator 
whose kernel includes the functions $\eta_{j,k}$ describing 
the eigenfunction $\psi$ in a 
neighbourhood of all \textbf{pair coalescence} points 
$x_j = x_k$, \ $j, k = 1, 2, \dots, N$,\ $j \not = k$. 
We emphasize that neither the points  
$x_j = 0, j = 1, 2, \dots, N$, nor 
the coalescence points of higher orders (e.g. $x_j = x_k = x_l$ with pair-wise distinct $j, k, l$)   
contribute to the asymptotics \eqref{eq:main1}.

The paper is organized as follows. In Section \ref{sect:main} we describe the representation of 
the function $\psi$ near the pair coalescence points (see \eqref{eq:locan} for the case $N=2$), 
state the main result in its complete form as Theorem \ref{thm:maincompl}, which includes 
the formula \eqref{eq:coeffA} for 
the coefficient $A$, 
and give the details of the factorization $\BOLG = \BPsi^*\BPsi$. 
Section \ref{sect:compact} contains 
necessary facts about compact operators, 
and it includes asymptotic formulas for spectra of integral operators with homogeneous kernels.  
Section \ref{sect:model} is focused on spectral asymptotics of the model integral operator that is instrumental to the case $N\ge 3$.   
Using the factorization $\BOLG = \BPsi^*\BPsi$, in Sections \ref{sect:factor} and \ref{sect:trim} 
the main Theorem \ref{thm:maincompl} is restated in terms of the operator $\BPsi$, see Theorem 
\ref{thm:gtopsi}. Here we also construct suitable approximations for $\BPsi$, 
to which one can apply the results of Sect. \ref{sect:model}. 
Section. \ref{sect:proofs} completes the proof of Theorem \ref{thm:gtopsi} and hence that 
of Theorem \ref{thm:maincompl}.

We conclude the introduction with some general notational conventions.  

\textit{Coordinates.} 
As mentioned earlier, we use the following standard notation for the coordinates: 
$\bx = (x_1, x_2, \dots, x_N)$,\ where $x_j\in \R^3$, $j = 1, 2, \dots, N$. 
In order to write 
formulas in a more compact and unified way, we sometimes use the notation 
$x_0 = 0$. 

The vector $\bx$ is often represented in the form  
\begin{align*}
\bx = (\hat\bx_j, x_j) \quad \textup{with}\quad   
\hat\bx_j = (x_1, x_2, \dots, x_{j-1}, x_{j+1},\dots, x_N)\in\R^{3N-3}, 
\end{align*}
for arbitrary $j = 1, 2, \dots, N$. Most frequently we use this notation with $j=N$, and  
write $\hat\bx = \hat\bx_N$, so that $\bx = (\hat\bx, x_N)$.

For $N\ge 3$ it is also useful to introduce the notation 
for $\bx$ with $x_j$ and $x_k$ taken out: 
\begin{align}\label{eq:xtilde}
\begin{cases}
\tilde\bx_{j, k} = (x_1, \dots, x_{j-1}, x_{j+1}, \dots, x_{k-1}, x_{k+1},\dots, x_N), \quad \textup{if}\  j <k,\\[0.2cm] 
\qquad\textup{and}\  \tilde\bx_{j, k} = \tilde\bx_{k, j},\quad \textup{if}\ j >k.
\end{cases}
\end{align}
If $j < k$, then we write $\bx = (\tilde\bx_{j, k}, x_j, x_k)$. For any $j\le N-1$ 
the vector $\hat\bx$ can be represented as $\hat\bx = (\tilde\bx_{j, N}, x_j)$.

The notation $B_R$ is used for the ball $\{x\in\R^3: |x| < R\}$.

\textit{Derivatives.} 
Let $\mathbb N_0 = \mathbb N\cup\{0\}$.
If $x = (x', x'', x''')\in \R^3$ and $m = (m', m'', m''')\in \mathbb N_0^3$, then 
the derivative $\p_x^m$ is defined in the standard way:
\begin{align*}
\p_x^m = \p_{x'}^{m'}\p_{x''}^{m''}\p_{x'''}^{m'''}.
\end{align*}  

\textit{Cut-off functions.} 
We systematically use the following smooth cut-off functions. Let 
\begin{align}\label{eq:sco}
\t\in\plainC\infty_0(\R),\quad \z(t) = 1-\t(t), 
\end{align}
be functions such that $0\le \t\le 1$ and 
\begin{align}\label{eq:sco1} 
\t(t) = 0,\quad \textup{if}\quad |t|>1;\ \quad
\t(t) = 1,\quad \textup{if}\quad |t|<\frac{1}{2}. \ 
\end{align}

\textit{Integral operators.} 
The notation $\iop(\CK)$ is used for the integral operator with kernel $\CK$, 
e.g. $\BOLG = \iop(\g)$.  
The functional spaces, where $\iop(\CK)$ acts are obvious from the context.

\textit{Bounds.} 
As explained earlier, for two non-negative numbers (or functions) 
$X$ and $Y$ depending on some parameters, 
we write $X\lesssim Y$ (or $Y\gtrsim X$) if $X\le C Y$ with 
some positive constant $C$ independent of those parameters.
To avoid confusion we often make explicit comments on the nature of 
(implicit) constants in the bounds.

\section{Representation formula. Details of the main result}\label{sect:main}

\subsection{Representation formula} 
Our approach is built on the sharp qualitative result for $\psi$ obtained in \cite{FHOS2009}. 
In order to write all the formulas in a more compact and unified way, we use the notation 
$x_0 = 0$. As before, $\bx = (x_1, x_2, \dots, x_N)\in\R^{3N}$. 
Thus, unless otherwise stated, the indices labeling the particles, run from $0$ to $N$. 

Denote 
\begin{align}\label{eq:sls}
\SfS_{l, s} = \{\bx\in\R^{3N}: x_l\not = x_s\},\ 
l\not = s. 
\end{align}
The function $\psi$ is real-analytic on the set 
\begin{align*}
\SU = 
\bigcap_{0\le l < s\le N}\SfS_{l, s}.
\end{align*}
For each pair $j, k: j\not = k$,  
we are interested in the behaviour of $\psi$ on the set 
\begin{align}\label{eq:uj}
\SU_{j,k} =  
\bigcap_{\substack{l \not = s\\
(l, s)\not = (j, k)}}  
\SfS_{l,s}.
\end{align}
In words, $\SU_{j, k}$ includes 
the coalescence point $x_j = x_k$, but excludes all the others. 
Our main focus will be on the function $\psi$ 
near the ``diagonal" set 
\begin{align}\label{eq:diag}
\SU^{(\rm d)}_{j,k} = \{\bx\in\SU_{j,k}: x_j = x_k\}.
\end{align}
The sets introduced above are obviously symmetric with respect to permutations of 
indices, e.g. $\SU_{j,k} = \SU_{k,j}$, $\SU^{(\rm d)}_{j, k}=\SU^{(\rm d)}_{k,j}$.
Observe also that the sets $\SU_{j,k}$, 
$\SU_{j,k}^{(\dc)}$ are of full measure in $\R^{3N}$ and $\R^{3N-3}$ respectively, 
and that they are connected.   

The following property follows from \cite[Theorem 1.4]{FHOS2009}.

\begin{prop}\label{prop:repr}
For each pair of indices $j, k = 0, 1, \dots, N$ such that $j\not = k$,  there exists 
an open connected set $\Om_{j,k} = \Om_{k, j}\subset \R^{3N}$, such that 
\begin{align}\label{eq:omin}
\SU_{j,k}^{(d)}\subset\Om_{j,k}\subset \SU_{j,k},
\end{align}
and two uniquely defined functions $\xi_{j, k}, \eta_{j, k}$, real analytic on $\Om_{j, k}$, such that 
for all $\bx\in \Om_{j, k}$  
the following representation holds:
\begin{align}\label{eq:repr}
 \psi(\bx) = \xi_{j, k}(\bx) + |x_j-x_k| \eta_{j, k}(\bx).
\end{align} 
\end{prop} 

Due to the uniqueness of functions $\xi_{j, k}, \eta_{j, k}$, we have the symmetry 
$\xi_{j, k} = \xi_{k, j}$, $\eta_{j, k} = \eta_{k, j}$ for all $j\not = k$. 
  
The asymptotic coefficient $A$ in the formula \eqref{eq:main} is defined via the functions 
$\eta_{j,k}$, $j, k = 1, 2, \dots, N, j<k$, on the sets \eqref{eq:diag}.  
Using the notation \eqref{eq:xtilde} we write the function $\eta_{j, k}(\bx)$ 
on $\SU_{j, k}^{(\dc)}$ as $\eta_{j, k}(\tilde\bx_{j, k}, x, x)$. 
As a by-product of the proof we obtain the following integrability properties. 

\begin{thm}\label{thm:etadiag} 
If $N\ge 3$, then each function $\eta_{j, k}(\ \cdot\ , x, x)$, $1\le j < k\le N$, 
belongs to $\plainL2(\R^{3N-6})$ for a.e. $x\in \R^3$ and   
the function 
\begin{align}\label{eq:H}
H(x):= \bigg[2 \sum\limits_{1\le j < k\le N}\int_{\R^{3N-6}} \big| 
\eta_{j, k}(\tilde\bx_{j,k}, x, x) \big|^2 d\tilde\bx_{j, k}\bigg]^{\frac{1}{2}},
\end{align} 
 belongs to $\plainL{\frac{3}{4}}(\R^{3})$. 
 
If $N = 2$, then the function $H(x):= \sqrt2 |\eta_{1, 2}(x, x)|$ 
belongs to $\plainL{\frac{3}{4}}(\R^{3})$.
\end{thm}
 
Having at our disposal this theorem, we can now state the main result of the paper 
in its complete form.
 
\begin{thm}\label{thm:maincompl}
Suppose that the eigenfunction $\psi$ satisfies the bound \eqref{eq:exp}. 
Then the eigenvalues $\l_k(\BOLG), k = 1, 2, \dots,$ 
of the operator $\BOLG$  satisfy the asymptotic formula \eqref{eq:main} with the constant 
\begin{align}\label{eq:coeffA}
A = \frac{1}{3}\bigg(\frac{2}{\pi}\bigg)^{\frac{5}{4}}
 \int_{\R^{3}} H(x)^{\frac{3}{4}} dx . 
\end{align} 
\end{thm}
  
\begin{rem} \label{rem:sym}
The coefficient $A$ can be equal to zero for some eigenfunctions $\psi$. For example, 
if we assume that 
the particles are spinless fermions, i.e. 
the function $\psi$ is antisymmetric, 
then it is immediate to see that 
for all $j, k$, $j\not = k$,
both components $\xi_{j, k}$ and $\eta_{j, k}$ in \eqref{eq:repr} 
vanish on the diagonal $\SU_{j, k}^{(\dc)}$, and as a consequence $A=0$. This means that 
$\l_k(\BOLG) = o(k^{-8/3})$. This fact can be interpreted by saying that antisymmetric 
eigenfunctions possess better than Lipschitz smoothness 
at the coalescence points, and hence the eigenvalues of $\BOLG$ decay faster. 

The fermionic nature of particles may 
manifest itself differently if we introduce 
the spin variable. 
In this case the antisymmetry of the \textit{full} eigenfunction comes either from the spatial 
component $\psi(\bx)$ or from the spin component. 
For illustration first consider the case of two electrons, i.e. $N=2$.  
In the \textit{triplet} configuration the antisymmetry is carried by the spatial component 
$\psi(x_1, x_2)$, see \cite[Subsect. 3.3.2]{HaKlKoTe2012}, 
and then, as pointed out a few lines above, we have $A = 0$.  
If the electrons 
are in the  \textit{singlet} configuration, then the spin component is antisymmetric, whereas    
the function $\psi = \psi(x_1, x_2)$ is symmetric 
and the diagonal value $\eta_{1, 2}(x, x)$ is not identically zero, 
see \cite[Subsect. 3.3.1]{HaKlKoTe2012}. Thus $A >0$.  

In the case $N\ge 3$ different electron pairs may form different configurations, in which case the triplet 
coalescences will not contribute to the coefficient $A$.  
\end{rem}

\begin{rem} 
If we assume that the function $\psi$ is symmetric or antisymmetric, then both the proof 
of the 
main asymptotic formula \eqref{eq:main}, and the formula   
\eqref{eq:H} can be simplified. 
Indeed, as we have seen, the factorization $\BOLG = \BPsi^*\BPsi$ holds with the 
simple looking integral operator $\BPsi$ given by \eqref{eq:ferbos}. This is in contrast with the general case, as will be evident from the next subsection. 
Furthermore, as discussed in Remark \ref{rem:sym}, for the antisymmetric $\psi$ we have $A = 0$. 
Assume that $N\ge 3$ and that $\psi$ is totally symmetric. 
It follows that for all 
$\tilde\by = (y_1, y_2, \dots, y_{N-2})\in \R^{3N-6}$,\ $x, t\in\R^3$ 
and $j\not =k, l\not = s$, we have
\begin{align*}
\psi(y_1, \dots, y_{j-1}, x, y_{j}, &\ \dots, y_{k-2}, t, y_{k-1},\dots, y_{N-2}) \\
= &\ \psi(y_1, \dots, y_{l-1}, x, y_{l}, \dots, y_{s-2}, t, y_{s-1},\dots, y_{N-2}). 
\end{align*} 
Due to the uniqueness of functions $\xi_{j, k}, \eta_{j, k}$ in Proposition \ref{prop:repr}, 
the above equality leads to 
\begin{align*}
\eta_{j, k}(y_1, \dots, y_{j-1}, x, y_{j}, &\ \dots, y_{k-2}, t, y_{k-1},\dots, y_{N-2}) \\
= &\ \eta_{l, s}(y_1, \dots, y_{l-1}, x, y_{l}, \dots, y_{s-2}, t, y_{s-1},\dots, y_{N-2}). 
\end{align*} 
As a consequence, the formula \eqref{eq:H} rewrites as
\begin{align*}
H(x)= \bigg[N(N-1) \int_{\R^{3N-6}} \big| 
\eta_{N-1, N}(\tilde\by, x, x) \big|^2\, d\tilde\by\bigg]^{\frac{1}{2}}.
\end{align*} 
\end{rem}
    
\subsection{Factorization of $\BOLG$: change of variables $(\hat\bx_j, x)\mapsto (\hat\bx, x)$} 
\label{subsect:fact}
In the general case (i.e. without any symmetry assumptions on $\psi$) 
the operator $\BPsi$ in the identity 
$\BOLG = \BPsi^*\BPsi$ looks more complicated compared to \eqref{eq:ferbos}. 
The purpose of this subsection is to describe 
this factorization 
%
%
and the associated change of variables.  
Rewrite the definition \eqref{eq:den} in the form:
\begin{align}\label{eq:psij}
\g(x, y) = &\ \sum_{j=1}^N \int_{\R^{3N-3}} \overline{\psi_j(\hat\bx, x)} \psi_j(\hat\bx, y) d\hat\bx,
\quad \textup{where} 
\notag\\
\psi_j(\hat\bx, x) = &\ \psi(x_1, \dots, x_{j-1}, x, x_j, \dots, x_{N-1}),\quad j = 1, 2, \dots, N.
\end{align}
Therefore $\boldsymbol\G$ can be represented as a product $\BOLG = \boldsymbol\Psi^*\boldsymbol\Psi$, 
where $\boldsymbol\Psi:\plainL2(\R^3)\to \plainL2(\R^{3N-3}; \mathbb C^N)$ 
is the integral operator with the vector-valued kernel 
\begin{align}\label{eq:bpsi}
\BPsi(\hat\bx, x) = \{\psi_j(\hat\bx, x)\}_{j=1}^N. 
\end{align}
As explained in the Introduction, given this factorization, 
the asymptotic relation \eqref{eq:main} translates to the 
formula \eqref{eq:main1}. Later we state this fact 
again as Theorem \ref{thm:gtopsi} using a more convenient notation. 

The change of variables $(\hat\bx_j, x)\mapsto (\hat\bx, x)$ plays an important role throughout 
the paper. In particular, it is crucial to recast 
Proposition \ref{prop:repr} in terms of the new variables $(\hat\bx, x)$, which is done below.  

Let $j\not = k$, and let $\Om_{j, k}$ be the 
sets and $\xi_{j, k}(\bx)$, $\eta_{j, k}(\bx)$ be the functions
from Proposition \ref{prop:repr}.
For all $j = 1, 2, \dots, N$ and all $k = 0, 1, \dots, N-1,$ denote
\begin{align*}
\tilde\Om_{j, k} = 
\begin{cases}
\{(\hat\bx, x)\in\R^{3N}: 
(x_1, \dots, x_{j-1}, x, x_{j}, \dots, x_{N-1})\in\Om_{j, k}\},\quad \textup{if}\ j\ge k+1,\\
\{(\hat\bx, x)\in\R^{3N}: 
(x_1, \dots, x_{j-1}, x, x_{j}, \dots, x_{N-1})\in\Om_{j, k+1}\},\quad \textup{if}\ j\le k.
\end{cases}
\end{align*}
According to \eqref{eq:omin} we have  
\begin{align}\label{eq:tom}
\SU_{N,k}^{(\dc)}\subset \tilde\Om_{j, k} \subset \SU_{N,k},
\end{align}  
for all $k\le N-1$ and $j = 1, 2, \dots, N$, $j\not = k$. 
Together with functions $\xi_{j,k}, \eta_{j,k}$ define 
\begin{align*}
\tilde\xi_{j, k}(\hat\bx, x) = 
\begin{cases}
\xi_{j,k}(x_1, \dots, x_{j-1}, x, x_{j}, \dots, x_{N-1}),\quad \textup{if}\ j\ge k+1,\\[0.2cm]
\xi_{j,k+1}(x_1, \dots, x_{j-1}, x, x_{j}, \dots, x_{N-1}),\quad \textup{if}\ j\le k,
\end{cases}
\end{align*}
and 
\begin{align}\label{eq:teta}
\tilde\eta_{j,k}(\hat\bx, x) = 
\begin{cases}
\eta_{j,k}(x_1, \dots, x_{j-1}, x, x_{j}, \dots, x_{N-1}),\quad \textup{if}\ j\ge k+1,\\[0.2cm]
\eta_{j,k+1}(x_1, \dots, x_{j-1}, x, x_{j}, \dots, x_{N-1}),\quad \textup{if}\ j\le k.
\end{cases}
\end{align}
By Proposition \ref{prop:repr}, 
for each $j = 1, 2, \dots, N$, and each $k = 0, 1, \dots, N-1$, we have 
\begin{align}\label{eq:trepr}
\psi_j(\hat\bx, x) = \tilde\xi_{j, k}(\hat\bx, x) + |x_k-x| \tilde\eta_{j, k}(\hat\bx, x),
\quad \textup{for all}\ (\hat\bx, x) \in \tilde\Om_{j, k}.
\end{align} 
Observe that the newly introduced sets $\tilde\Om_{j, k}$ 
and the functions $\tilde\xi_{j,k}, \tilde\eta_{j,k}$  are 
not symmetric under the permutation $j\leftrightarrow k$.

The function \eqref{eq:H} can be easily rewritten via the new functions $\tilde\eta_{j, k}$:
\begin{align}\label{eq:tH}
H(x) = 
\begin{cases}
\big(|\tilde\eta_{1, 1}(x, x)|^2 + |\tilde\eta_{2, 1}(x, x)|^2\big)^{1/2},\quad \textup{if}\ N = 2;\\[0.3cm]
\bigg[\sum\limits_{j=1}^N\sum\limits_{k=1}^{N-1}\int_{\R^{3N-6}} \big| 
\tilde\eta_{j, k}(\tilde\bx_{k,N}, x, x) \big|^2 d\tilde\bx_{k, N}\bigg]^{\frac{1}{2}},\quad 
\textup{if}\ N\ge 3.
\end{cases}
\end{align} 
For $N=2$ the above formula is a consequence of the symmetry relation 
$\eta_{1, 2} = \eta_{2, 1}$ and equalities 
$\eta_{1,2}(x, x) = \tilde\eta_{1, 1}(x, x)$, 
$\eta_{2,1}(x, x) = \tilde\eta_{2,1}(x, x)$,  
 which follow from the definition \eqref{eq:teta}.
 
Now assume that $N\ge 3$.   
In view of the symmetry 
$\eta_{j,k}=\eta_{k, j}$ we can 
rewrite \eqref{eq:H} extending the summation to all $j, k$ such that $j\not = k$:
\begin{align*}
H(x)^2 = \sum_{j=1}^{N-1} \sum_{k = j+1}^N \int_{\R^{3N-6}} \big| 
\eta_{j, k}(\tilde\bx_{j,k}, x, x) \big|^2 d\tilde\bx_{j, k}
+ \sum_{j=1}^N \sum_{k = 1}^{j-1} \int_{\R^{3N-6}} \big| 
\eta_{j, k}(\tilde\bx_{j,k}, x, x) \big|^2 d\tilde\bx_{j, k}. 
\end{align*} 
By \eqref{eq:teta}, the second sum coincides with
\begin{align*}
\sum_{j=1}^N \sum_{k = 1}^{j-1} \int_{\R^{3N-6}} \big| 
\tilde\eta_{j, k}(\tilde\bx_{k,N}, x, x) \big|^2 d\tilde\bx_{k, N},
\end{align*} 
and the first one coincides with
\begin{align*}
\sum_{j=1}^{N-1} \sum_{k = j+1}^N \int_{\R^{3N-6}} \big| 
\tilde\eta_{j, k-1}(\tilde\bx_{k-1, N}, x, x) \big|^2 d\tilde\bx_{k-1, N}
= \sum_{j=1}^{N-1} \sum_{k = j}^{N-1} \int_{\R^{3N-6}} \big| 
\tilde\eta_{j, k}(\tilde\bx_{k, N}, x, x) \big|^2 d\tilde\bx_{k, N}.
\end{align*}
Adding the first and second sums together we obtain \eqref{eq:tH}, as claimed.

\section{Compact operators}\label{sect:compact}
   
\subsection{Compact operators} For information on compact operators we use mainly  
Chapter 11 of the book \cite{BS}, where one can also find further references. 
Let $\CH$ and $\mathcal G$ be separable Hilbert spaces.  
Let $T:\CH\to\mathcal G$ be a compact operator. 
If $\CH = \mathcal G$ and $T=T^*\ge 0$, then $\l_k(T)$, $k= 1, 2, \dots$, 
denote the positive eigenvalues of $T$ 
numbered in descending order counting multiplicity. 
For arbitrary spaces $\CH$, $\mathcal G$ and compact $T$, by $s_k(T) >0$, 
$k= 1, 2, \dots$, we denote the singular values of 
$T$ defined by $s_k(T)^2 = \l_k(T^*T) = \l_k(TT^*)$. 
%

We classify compact operators by the rate of decay of their singular values. 
If $s_k(T)\lesssim k^{-1/p}, k = 1, 2, \dots$, 
with some $p >0$, then we say that $T\in \BS_{p, \infty}$ and denote
\begin{align*}
\| T\|_{p, \infty} = \sup_k s_k(T) k^{\frac{1}{p}}.
\end{align*}
These classes are discussed in detail in \cite[\S 11.6]{BS}. 
The class $\BS_{p, \infty}$ is a complete linear space 
with the quasi-norm $\|T\|_{p, \infty}$.
For all $p>0$ the quasi-norm satisfies the following ``triangle" inequality for  
operators $T_1, T_2\in\BS_{p, \infty}$:
\begin{align}\label{eq:triangle}
\|T_1+T_2\|_{\scalet{p}, \scalel{\infty}}^{{\frac{\scalel{p}}{\scalet{p+1}}}}
\le \|T_1\|_{\scalet{p}, \scalel{\infty}}^{{\frac{\scalel{p}}{\scalet{p+1}}}}
+ \|T_2\|_{\scalet{p, \infty}}^{{\frac{\scalel{p}}{\scalet{p+1}}}}.
\end{align} 
For $T\in \BS_{p, \infty}$ the following numbers are finite:
\begin{align}\label{eq:limsupinf}
\begin{cases}
\SfG_p(T) = 
\big(\limsup\limits_{k\to\infty} k^{\frac{1}{p}}s_k(T)\big)^{p}
= \limsup\limits_{s\to 0} s^p n(s, T),\\[0.3cm]
\sg_p(T) =  
\big(\liminf\limits_{k\to\infty} k^{\frac{1}{p}}s_k(T)\big)^{p} = \liminf\limits_{s\to 0} s^p n(s, T),
\end{cases}
\end{align}
and they clearly satisfy the inequalities
\begin{align*} 
\sg_p(T)\le \SfG_p(T)\le \|T\|_{p, \infty}^p.
\end{align*}
Note that $\SfG_q(T) = 0$ for all $q > p$. Observe that 
\begin{align}\label{eq:double}
\sg_{p}(T T^*) = \sg_{p}(T^*T) = \sg_{2p}(T),\quad 
\SfG_{p}(T T^*) = \SfG_{p}(T^*T) = \SfG_{2p}(T).
\end{align}
If $\SfG_p(T) = \sg_p(T)$, then the singular values of $T$ satisfy the asymptotic formula
\begin{align*} 
s_n(T) = \big(\SfG_p(T)\big)^{\frac{1}{p}} n^{-\frac{1}{p}} + o(n^{-\frac{1}{p}}),\ n\to\infty.
\end{align*}
The functionals $\sg_p(T)$, $\SfG_p(T)$ also satisfy the inequalities of the type \eqref{eq:triangle}:
\begin{align}\label{eq:triangleg}
\begin{cases}
\SfG_p(T_1+T_2)^{{\frac{\scalel{1}}{\scalet{p+1}}}}
\le \SfG_p(T_1)^{{\frac{\scalel{1}}{\scalet{p+1}}}}
+ \SfG_p(T_2)^{{\frac{\scalel{1}}{\scalet{p+1}}}},\\[0.3cm]
\sg_p(T_1+T_2)^{{\frac{\scalel{1}}{\scalet{p+1}}}}
\le \sg_p(T_1)^{{\frac{\scalel{1}}{\scalet{p+1}}}}
+ \SfG_p(T_2)^{{\frac{\scalel{1}}{\scalet{p+1}}}}.
\end{cases}
\end{align}
It  follows from these inequalities that 
the functionals $\SfG_p$ and $\sg_p$ are continuous on $\BS_{p, \infty}$:
\begin{align*}
\big|
\SfG_p(T_1)^{{\frac{\scalel{1}}{\scalet{p+1}}}} - \SfG_p(T_2)^{{\frac{\scalel{1}}{\scalet{p+1}}}}
\big|\le &\ \SfG_p(T_1-T_2)^{{\frac{\scalel{1}}{\scalet{p+1}}}},\\
\big|\sg_p(T_1)^{{\frac{\scalel{1}}{\scalet{p+1}}}} - \sg_p(T_2)^{{\frac{\scalel{1}}{\scalet{p+1}}}}
\big|\le &\ \SfG_p(T_1-T_2)^{{\frac{\scalel{1}}{\scalet{p+1}}}}.
\end{align*}  
We need the following two corollaries of this fact:

\begin{cor}\label{cor:zero}
Suppose that $\SfG_p(T_1-T_2) = 0$. Then 
\begin{align*}
\SfG_p(T_1) = \SfG_p(T_2),\quad \sg_p(T_1) = \sg_p(T_2).
\end{align*}
\end{cor}

The next corollary is more general:

\begin{cor}\label{cor:zero1}
Suppose that $T\in\BS_{p, \infty}$ and that for every $\nu>0$ there exists an operator 
$T_\nu\in \BS_{p, \infty}$ such that $\SfG_p(T - T_\nu)\to 0$, 
$\nu\to 0$. Then the functionals 
$\SfG_p(T_\nu), \sg_p(T_\nu)$ have limits as $\nu\to 0$ and 
\begin{align*}
\lim_{\nu\to 0} \SfG_p(T_\nu) = \SfG_p(T),\quad 
\lim_{\nu\to 0} \sg_p(T_\nu) = \sg_p(T).
\end{align*}
\end{cor}

 \subsection{Estimates for singular values of integral operators}
The final ingredients of the proof are the 
results due to M.S. Birman and M.Z. Solomyak, 
investigating the membership of integral operators in various classes 
of compact operators.

For estimates of the 
singular values we rely on \cite[Corollaries 4.2, 4.4, Theorem 4.4]{BS1977}, 
which we state here in a form convenient for 
our purposes.  Below we use the following notation which is standard in the theory of 
Sobolev spaces: 
$\plainH{l}(\R^d) = \plainW{2,l}(\R^d)$. 

\begin{prop}\label{prop:BS} 
Let $a\in\plainL\infty(\R^d)$, $b\in\plainL2_{\textup{\tiny{\rm loc}}}(\R^n)$. 
Assume that the function $a$ has compact support. 
Suppose that 
$T(t, x)$, $t\in\R^n$, $x\in\R^d$, is a kernel such that 
$T(t, \ \cdot\ )\in \plainH{l}(\R^d)$ with some $l = 0, 1, \dots$, 
for a.e. $t\in\R^n$, and the function $\| T(t, \ \cdot\ )\|_{\plainH{l}}$ is in 
$\plainL2(\R^n, |b(t)|^2 dt)$.

Let $T_{ba}: \plainL2(\R^d)\to\plainL2(\R^n)$, 
be the integral operator 
\begin{align*}
(T_{ba}u)(t) = b(t) \int T(t, x) a(x) u(x)\,dx,\quad u\in\plainL2(\R^d).
\end{align*}
Then 
\begin{align}\label{eq:spq}
\sum_{k=0}^\infty k^{\frac{2l}{d}} s_k(T_{ba})^2 < \infty,
\end{align}
and hence $s_k(T_{ba}) = o(k^{-1/q})$, where $1/q = 1/2+l/d$. In other words, 
$\SfG_q(T_{ba}) = 0$. 
\end{prop}

The original results in \cite[Corollaries 4.2, 4.4, Theorem 4.4]{BS1977} 
are considerably more general and more 
precise: instead of just the finiteness statement \eqref{eq:spq}, they contain estimates depending explicitly on the kernel $T$ and weights $a, b$. These estimates have slightly 
different form for different cases $2l >d, 2l = d$ and $2l <d$, 
and therefore, to avoid cumbersome formulations we chose not to quote them in detail.  

The next group of results is concerned with spectral asymptotics for 
integral operators.

\subsection{Integral operators with homogeneous kernels} 
First we consider pseudo - differential operators with asymptotically homogeneous matrix-valued symbols. 
Spectral asymptotics for such operators were studied in \cite{BS1977_1}, \cite{BS1979}. In fact, 
these papers allow for more general operators, but we need only a relatively simple 
special case of those results.  
Precisely, let $\CA(x), \CB(x), X(\xi)$, where $x, y, \xi\in\R^d$, be 
rectangular matrix-valued functions 
of matching dimensions, so that the product 
\begin{align}\label{eq:sy}
\CB(x) X(\xi)\CA(y) 
\end{align}
is again a rectangular matrix.
Assume that 
\begin{align}\label{eq:matr}
\CB\in \plainC{}_0(\R^d),
\quad \CA\in\plainC{}_0(\R^d). 
\end{align}
We do not reflect the matrix nature of the functional spaces in the 
notation to avoid cumbersome formulas, and 
this should not cause confusion. 
Suppose that $X(\xi)$ is a bounded function which is 
asymptotically homogeneous of negative order, i.e. 
there exists a matrix-valued function $X_\infty\in\plainC\infty(\R^d\setminus\{0\})$ such that for some 
$\tau >0$,
\begin{align}\label{eq:homo}
X_\infty(t\xi) = t^{-\tau}X_\infty(\xi), \ \xi\not = 0, 
\end{align}
for all $t >0$, and 
\begin{align}\label{eq:Xas}
X(\xi) - X_\infty(\xi) = o(|\xi|^{-\tau}),\quad |\xi|\to \infty.
\end{align}
Define the matrix-valued function
\begin{align*}
\CT_\infty(x, \xi) = \CB(x) X_\infty(\xi)\CA(x). 
\end{align*}

\begin{prop}\label{prop:homo}
Let the above conditions on 
$\CA, \CB, X$ be satisfied and let $p = d\tau^{-1}$. 
Then the pseudo-differential operator $T:\plainL2(\R^d)\to\plainL2(\R^d)$ 
defined by the formula 
\begin{align}\label{eq:pdo}
(Tu)(x) = \frac{1}{(2\pi)^{d}}\iint\CB(x) e^{i\xi(x-y)}
X(\xi) \CA(y) u(y) dy d\xi,
\end{align}
is compact, it belongs to $\BS_{p, \infty}$ and satisfies the asymptotic formula 
\begin{align}\label{eq:pdoas}
\SfG_p(T) = \sg_p(T) = 
\frac{1}{d (2\pi)^d} \int\limits_{\R^d} \int\limits_{\mathbb S^{d-1}} 
\sum\limits_k \big[s_k\big( \CT_\infty(x, \om)\big)\big]^p\, d\om dx.
\end{align}
%
%
\end{prop}

This proposition is a consequence of Theorem 2 from \cite{BS1977_1} and Remark 3 following 
this theorem.

We apply Proposition \ref{prop:homo} 
to  integral operators with homogeneous kernels. Let  
$\Phi\in\plainC\infty(\R^d\setminus\{0\})$ be a matrix-valued function such that 
\begin{align}\label{eq:homophi}
\Phi(tx) = t^{\a} \Phi(x),\quad x\not = 0, \a >-d, 
\end{align} 
for all $t >0$. Consider the integral operator $W$ with the matrix-valued kernel
\begin{align*}
W(x, y) = \CB(x) \Phi(x-y)  \CA(y)
\end{align*} 
with $\CA, \CB$ satisfying the conditions \eqref{eq:matr}, and 
assuming that the matrix dimensions are matched in the same way as for the 
symbol \eqref{eq:sy}. 
We study spectral asymptotics of the operator $W$ by reducing it to the operator of the form 
\eqref{eq:pdo}. 
Let $\t$ be as defined in \eqref{eq:sco}, \eqref{eq:sco1}, and let 
$R_0>0$ be a number such that 
\begin{align*}
W(x, y) = W(x, y) \t\big(|x-y| R^{-1}\big), \quad \textup{for all}\quad R\ge R_0.
\end{align*}
Consequently, the operator $W$ has the form 
\eqref{eq:pdo} with the function  
\begin{align}\label{eq:XR}
X(\xi) = X_R(\xi) =  \int e^{-i\xi x} \t\big(|x| R^{-1}\big) \Phi(x) dx.
\end{align} 
Integrating by parts, we conclude that for each $\xi\not = 0$ 
the function $X_R(\xi)$ converges as $R\to\infty$ to a $\plainC\infty(\R^{d}\setminus\{0\})$-function
\begin{align}\label{eq:X0}
X_\infty(\xi) = \lim_{R\to\infty}X_R(\xi).
\end{align}
The function $X_\infty$ satisfies \eqref{eq:homo} with $\tau = \a + d$. Indeed, 
using \eqref{eq:homophi} write for $t >0$:
\begin{align}\label{eq:erelimit}
X_R(t\xi) = &\ t^{-\a-d}  
\int e^{-i\xi x} \t\big(|x|(Rt)^{-1}\big) \Phi(x) dx\notag\\
= &\ t^{-\a-d} X_{Rt}(\xi). 
\end{align}
Passing to the limit as $R\to\infty$, we get \eqref{eq:homo} with $\tau = \a + d$, as claimed. 
The equality \eqref{eq:erelimit} also implies that
\begin{align*}
X_R(t\xi) - X_\infty(t\xi) = t^{-\a-d}\big(X_{Rt}(\xi) - X_\infty(\xi)\big) = o(t^{-\a-d}), 
\quad t\to\infty, 
\end{align*}
for each $\xi\in\R^d$ and $R>0$, which entails \eqref{eq:Xas}. Thus, applying 
Proposition \ref{prop:homo}, we obtain the spectral asymptotics for the operator $W$. 

\begin{cor} \label{cor:w}
The operator $W$ has the form \eqref{eq:pdo} with the function $X\in \plainC\infty(\R^d)$ 
defined in \eqref{eq:XR}. The singular values of $W$ satisfy the relation 
\eqref{eq:pdoas} with $p^{-1} = 1+\a d^{-1}$.
\end{cor}

We need this result for the special case of scalar $\CA=a\in \plainC{}_0(\R^d), 
\CB=b\in \plainC{}_0(\R^d)$, and   
\begin{align}\label{eq:BX}
\Phi(x) =  \{\phi_j(x)\}_{j=1}^m
\end{align}
with scalar $\a$-homogeneous functions $\phi_j$, $j = 1, 2, \dots, m$. 
As the next assertion shows, 
in this case the right-hand side of \eqref{eq:pdoas} can be easily evaluated.

\begin{cor}\label{cor:vec}
Suppose that $\Phi$ is given as in \eqref{eq:BX} with some $\a$-homogeneous 
scalar functions $\phi_j$, $j = 1, 2, \dots, k,$ with $\a > -d$.  
Then 
\begin{align*}
\SfG_p(W) = \sg_p(W) = \frac{1}{d(2\pi)^d}
\int_{\mathbb S^{d-1}}|X_\infty(\om)|^p\, d\om 
\int_{\R^d}|a(x)b(x)|^p \, dx, 
\end{align*}
where $p^{-1} = 1 + \a d^{-1}$.
\end{cor}

\begin{proof}
The matrix $\CT_\infty(x, \xi)$ is rank one and 
\begin{align*}
s_1\big(\CT_\infty(x, \xi)\big) = |a(x)|\,|b(x)|\, |X_\infty(\xi)|.
\end{align*}
The required formula follows from Corollary \ref{cor:w}.
\end{proof}

Consider two examples in which the above formula can be simplified further. 
The first example is crucial for the proof of Theorem \ref{thm:maincompl}. 

\begin{example}\label{ex:scal}
Let $m = 1$, and let 
$\Phi(x) = \phi(x) = |x|^{\a}$, $\a > -d$, 
be a scalar function.  
Then (see, e.g. \cite[Ch. 2, Sect. 3.3]{GS1964})
\begin{align*}
X_\infty(\xi) = 2^{d+\a}\pi^{\frac{d}{2}} 
\frac{\G\big(\frac{d+\a}{2}\big)}{\G\big(-\frac{\a}{2}\big)}|\xi|^{-(d+\a)},\quad \a \not = 0, 2, 4, \dots, 
\end{align*}
and $X_\infty(\xi) = 0$ for $\a = 0, 2, 4, \dots$. 
Thus, for $1/p = 1+\a/d$ and $\a\not = 0, 2, 4, \dots,$ we have 
\begin{align}\label{eq:Xint}
\mu_{\a, d}:= \frac{1}{d(2\pi)^d}\int_{\mathbb S^{d-1}} 
|X_\infty(\om)|^p d\om 
= \bigg[\frac{\G\big(\frac{d+\a}{2}\big)}{\pi^{\frac{\a}{2}}|\G\big(-\frac{\a}{2}\big)|}\bigg]^p 
\frac{1}{\G\big(\frac{d}{2}+1\big)}.
\end{align}
Now Corollary \ref{cor:vec} yields 
\begin{align*}
\SfG_p(W) = \sg_p(W) = \mu_{\a, d}
\int_{\R^d} 
|a(x)b(x)|^p  dx,\quad \frac{1}{p} = 1 + \frac{\a}{d}.
\end{align*}  
\end{example}

Note that the case of scalar functions $\Phi$ was studied in \cite{BS1970}, 
see also \cite[Theorem 10.9]{BS1977}. 
Next we consider an important example of a vector-valued function $\Phi$. 
We do not need it for the current paper but prepare it for future use.

\begin{example}\label{ex:grad}
Let $m = d$, and let $\Phi(x) = \nabla |x|^{\a+1} 
= (\a+1)|x|^{\a-1}x$,\ $\a > -d$.  This vector-valued 
function is homogeneous of order $\a$ and  (similarly to \cite[Ch. 2, Sect. 3.3]{GS1964})
\begin{align*}
X_\infty(\xi) = -i(\a+1) 2^{d+\a}\pi^{\frac{d}{2}}
\frac{\G\big(\frac{d+\a+1}{2}\big)}{\G\big(\frac{-\a+1}{2}\big)}
|\xi|^{-(\a+1+d)} \xi, \quad \a \not = 1, 3, 5, \dots,
\end{align*}
and $X_\infty(\xi) = 0$ for $\a = 1, 3, 5, \dots$. 
Thus, for $1/p = 1+\a/d$ and $\a\not = 1, 3, 5, \dots,$ we have 
\begin{align}\label{eq:Xintvec}
\nu_{\a, d}:= \frac{1}{d(2\pi)^d}\int_{\mathbb S^{d-1}} 
|X_\infty(\om)|^p d\om 
=  \bigg[\frac{(\a+1)
\G\big(\frac{d+\a+1}{2}\big)}{\pi^{\frac{\a}{2}}|\G\big(\frac{-\a+1}{2}\big)|}\bigg]^p 
\frac{1}{\G\big(\frac{d}{2}+1\big)}.
\end{align}
Now Corollary \ref{cor:vec} yields 
\begin{align*}
\SfG_p(W) = \sg_p(W) = \nu_{\a, d}
\int_{\R^d} |a(x)b(x)|^p dx,\quad \frac{1}{p} = 1 + \frac{\a}{d}.
\end{align*}
\end{example} 

\section{Spectral asymptotics for the model problem}\label{sect:model}

The objective of this section is to find the spectral asymptotics for a model 
integral operator. Recall that for any function $\mathcal K = \mathcal K(x, y)$, $x\in \R^n, y\in\R^d$, 
we denote by $\iop(\CK)$ the integral operator acting from $\plainL2(\R^d)$ into $\plainL2(\R^n)$.  
In each case the values of $n$ and $d$ are clear from the context. 
If $\CK(x, y)$ is $\mathbb C^s$-valued then the ``target" space 
$\plainL2(\R^n)$ is replaced by $\plainL2(\R^n; \mathbb C^s)$.

\subsection{The model operator} 
Let $a, b_{j,k}, \b_{j,k}$, $j = 1, 2, \dots, N$, $k = 1, 2, \dots, N-1$, be 
scalar functions such that 
\begin{align}\label{eq:abbeta}
\begin{cases}
a\in\plainC{\infty}_0(\R^3), &\ \quad b_{j,k}\in \plainC\infty_0(\R^{3N-3}),\\[0.2cm] 
\b_{j,k}\in\plainC\infty(\R^{3N}), 
\end{cases}
\end{align}
for all $j = 1, 2, \dots, N$, $k = 1, 2, \dots, N-1$. 
Let $\Phi\in\plainC\infty(\R^3\setminus \{0\})$ be a vector-valued  
function with $m$ scalar components, 
homogeneous of order $\a>-3$, as defined in \eqref{eq:BX}. 
Consider the vector-valued kernel $\CM(\hat\bx, x)$ with $mN$ components: 
\begin{align}\label{eq:cm}
\begin{cases}
\CM(\hat\bx, x) = \{\CM_j(\hat\bx, x)\}_{j=1}^N,\ 
\quad
\CM_j(\hat\bx, x) = \sum_{k=1}^{N-1} \CM_{j,k}(\hat\bx, x),\\[0.3cm]
\CM_{j,k}(\hat\bx, x) =  b_{j,k}(\hat\bx) \Phi(x_k-x) a(x)\b_{j,k}(\hat\bx, x). 
\end{cases}
\end{align}
Our aim is to find an asymptotic formula for the singular values of the operator 
$\iop(\CM): \plainL2(\R^3)\to\plainL2(\R^{3N-3}; \mathbb C^{mN})$. 
Although the function $\Phi(x)$ is homogeneous, 
the results on homogeneous kernels, notably Corollary \ref{cor:vec}, 
are not applicable directly, since 
the number of ``target" variables (i.e. $3N-3$) 
is greater than the number of the input variables (i.e. $3$), unless $N=2$. 
The proof of Theorem \ref{thm:model} below amounts to reducing the operator 
$\iop(\CM)$ to a form for which Corollary \ref{cor:vec} can be used. 
Recall that the weights $\CA$ and $\CB$ in Corollary \ref{cor:vec} are only required to be continuous 
(with compact support). Thus the smoothness restrictions on the functions $a$, $b_{j, k}$ 
$\b_{j, k}$ in the definition \eqref{eq:cm} can be relaxed, but for our 
purposes it suffices to assume conditions \eqref{eq:abbeta}. Moreover, this assumption allows us to 
avoid unnecessary technical complications. 
 
We use the representations $(\hat\bx, x) = (\tilde\bx_{k, N}, x_k, x)$ introduced in \eqref{eq:xtilde}.  
Denote  
\begin{align}\label{eq:hm}
\begin{cases}
h(t) = \bigg[\sum_{j=1}^N\sum_{k=1}^{N-1}\int_{\R^{3N-6}} 
| b_{j,k}(\tilde\bx_{k, N}, t)\b_{j,k}(\tilde\bx_{k, N}, t, t) |^2 d\tilde\bx_k\bigg]^{\frac{1}{2}},\ 
\textup{if}\ N\ge 3;\\[0.3cm]
h(t) = \big(|b_{1,1}(t)\b_{1,1}(t, t)|^2 
+ |b_{2,1}(t)\b_{2,1}(t, t)|^2\big)^{\frac{1}{2}},\ \textup{if}\ N=2.
\end{cases}
\end{align}
Let $X_\infty(\xi), \xi\in\R^3$, be the function defined by \eqref{eq:XR} and \eqref{eq:X0}.  
 
\begin{thm}\label{thm:model} 
Let $\CM$ be the operator defined above, where   
$\Phi\in\plainC\infty(\R^3\setminus\{0\})$ is a homogeneous vector 
function of order $\a > -5/2$.  
Then the operator $\iop(\CM)$ belongs to $\BS_{p, \infty}$, $1/p = 1+ \a/3$,  and 
\begin{align}\label{eq:model} 
\SfG_p\big(
\iop(\CM)\big)
= \sg_p\big(\iop(\CM)\big) 
= \frac{1}{24 \pi^3} 
\int_{\mathbb S^2}|X_\infty(\om)|^p\,d\om \int_{\R^3}
\big(|a(x) h(x)|\big)^p\,dx.
\end{align}
\end{thm} 
 
Throughout the proof we assume that $N\ge 3$. For $N=2$ the argument simplifies, and we 
omit it. 

We begin the proof with the following lemma. 

\begin{lem}\label{lem:cross} 
For each $j = 1, 2, \dots, N$ and each pair $k, l = 1, 2, \dots, N-1$, $k\not = l$, we have 
\begin{align*}
\SfG_{p/2}\big(\iop(\CM_{j,k})^*\iop(\CM_{j,l})\big) = 0.
\end{align*}
\end{lem}

\begin{proof} 
Fix a $j = 1, 2, \dots, N$ and write the kernel of the operator 
$\iop(\CM_{j,k})^*\iop(\CM_{j,l})$:
\begin{align*}
\CP_{k,l}(x, y) 
= &\ \overline{a(x)} a(y)\int\overline{\CM_{j,k}(\hat\bx, x)} \CM_{j,l}(\hat\bx, y) d\hat\bx \\
= &\ \overline{a(x)} a(y) \int \overline{\Phi(x-x_k)}\cdot \Phi(x_l-y) \,
\overline{b_{j,k}(\hat\bx)\b_{j,k}(\hat\bx, x)} b_{j,l}(\hat\bx) \b_{j,l}(\hat\bx, y)\,d\hat\bx.
\end{align*} 
Write $\hat\bx = (\tilde\bx_{l, N}, x_l)$, $d\hat\bx = d\tilde\bx_{l, N} dx_l$  
and change $x_l$ to $x_l+y$, so that
\begin{align*}
\CP_{k,l}(x, y) =  \overline{a(x)} a(y) \int \overline{\Phi(x-x_k)}\cdot &\ \Phi(x_l)   
\overline{b_{j,k}(\tilde\bx_{l, N}, x_l+y)\b_{j,k}(\tilde\bx_{j, N}, x_j+y, x)} \\
&\qquad\qquad \times b_{j,l}(\tilde\bx_{l, N}, x_l+y) 
\b_{j,l}(\tilde\bx_{l, N}, x_l+y, y) d\tilde\bx_{l, N} dx_l.
\end{align*}
Because of the conditions \eqref{eq:abbeta} for all $x\in \R^3$ 
the kernel $\CP_{k,l}$ is a $\plainC\infty_0$-function of $y\in\R^3$. 
Hence by Proposition \ref{prop:BS} the singular values of the operator $\iop(\CP_{k,l})$ decay 
faster than any negative power of their number.  In particular, 
$\SfG_{p/2}\big(\iop(\CP_{k,l})\big) = 0$, as required.
\end{proof}

\subsection{Proof of Theorem \ref{thm:model} for $\b_{jk} = 1$} \label{subsect:beta1} 
First we prove Theorem \ref{thm:model} for the simpler case $\b_{j,k} = 1$.
It follows from Lemma \ref{lem:cross} and from the inequality \eqref{eq:triangleg} that 
\begin{align*}
\SfG_{p/2}\bigg(\sum_{j=1}^N\sum_{k\not = l} \iop(\CM_{j,k})^*\iop(\CM_{j,l})\bigg) = 0.
\end{align*}
By Corollary \ref{cor:zero} this implies that 
\begin{align*}
\SfG_{p/2}\big(\iop(\CM)^* \iop(\CM)\big) 
= \SfG_{p/2}\bigg(\sum_{j=1}^N
\sum_{k=1}^{N-1} \iop(\CM_{j,k})^*\iop(\CM_{j,k})\bigg),
\end{align*}
and the same equality holds for the functional $\sg_{p/2}$.
Let us write the kernel 
$\CF(x, y)$ 
of the operator on the right-hand side, remembering that $\b_{j,k} = 1$:
\begin{align*}
\CF(x, y) = \overline{a(x)}a(y)&\ \sum_{j=1}^N
\sum_{k=1}^{N-1} \int \overline{\Phi(x-x_k)}\cdot \Phi(x_k-y) 
|b_{j,k}(\hat\bx)|^2
d\hat\bx\\
= &\ \overline{a(x)}a(y)\sum_{j=1}^N
\sum_{k=1}^{N-1} 
\int_{\R^3} \overline{\Phi(x-t)}\cdot \Phi(t-y)  \int_{\R^{3N-6}}
|b_{j,k}(\tilde\bx_{k, N}, t)|^2
d\tilde\bx_{k, N} dt\\
= &\ \overline{a(x)}a(y)\int_{\R^3}  \overline{\Phi(x-t)}\cdot \Phi(t-y) \ h(t)^2dt,  
\end{align*}
where the function
\begin{align*}
h(t) = \bigg[\sum_{j=1}^N
\sum_{k=1}^{N-1} \int_{\R^{3N-6}} 
|b_{j,k}(\tilde\bx_{k, N}, t)|^2 d\tilde\bx_{k, N}\bigg]^{\frac{1}{2}}
\end{align*}
coincides with \eqref{eq:hm} for $\b_{j,k}=1$. 
Define the vector-valued kernel $\mathcal G$ by  
\begin{align*}
\mathcal G(x, y) = h(x)\Phi(x-y)a(y), 
\end{align*}
so that 
$\iop(\CF) = \iop(\mathcal G)^* \iop(\mathcal G)$.  
Thus the functionals $\SfG_{p/2}$ for the operators $\iop(\CM)^*\iop(\CM)$ and 
$\iop(\mathcal G)^* \iop(\mathcal G)$ coincide 
with each other, and the same applies to the functionals $\sg_{p/2}$.   
Consequently, by virtue of \eqref{eq:double}, 
\begin{align}\label{eq:mtog}
\SfG_{p}\big(\iop(\CM)\big) = \SfG_{p}\big(\iop(\mathcal G)\big),\quad
\sg_{p}\big(\iop(\CM)\big) = \sg_{p}\big(\iop(\mathcal G)\big). 
\end{align}
Since $b_{j, k}\in\plainC\infty_0$, the function $h$ belongs to $\plainC{}_0$. 
Thus, to find $\SfG_p$ and $\sg_p$ for the operator 
$\iop(\mathcal G)$ we can 
apply Corollary \ref{cor:vec} with $d=3$ and with 
the weights $b = h\in \plainC{}_0$ and  
$a\in\plainC{\infty}_0$, which gives 
\begin{align*} 
\SfG_p\big(
\iop(\mathcal G)\big)
= \sg_p\big(\iop(\mathcal G)\big) 
= \frac{1}{24 \pi^3} \int_{\R^3}\int_{\mathbb S^2}\big(|a(x) h(x)| |X_\infty(\om)|\big)^p\,d\om dx,
\end{align*}
with $1/p = 1+ \a/3$. 
By \eqref{eq:mtog}, this equality implies \eqref{eq:model}, which completes 
the proof of Theorem \ref{thm:model} for $\b_{j,k} = 1$.

\subsection{Proof of Theorem \ref{thm:model} for arbitrary $\b_{j,k}\in\plainC\infty$} 
We reduce the general case to the one considered in 
Subsect. \ref{subsect:beta1}.  Since $b_{j,k}$ and $a$ are compactly supported, 
without loss of generality we may assume that $\b_{j,k}\in\plainC\infty_0(\R^{3N})$.
For each $j = 1, 2, \dots, N$ represent 
\begin{align*}
\CM_j = \CA_j+\sum_{k=1}^{N-1}\CF_{j,k},
\end{align*}
where
\begin{align*}
\CA_j(\hat\bx, x) = &\ \sum_{k=1}^{N-1}\Phi(x_k-x) b_{j,k}(\hat\bx) a(x) \b_{j,k}(\hat\bx, x_j),\\
\CF_{j,k}(\hat\bx, x) = &\ \Phi(x_k-x) b_{j,k}(\hat\bx) 
a(x) \big(\b_{j,k}(\hat\bx, x) - \b_{j,k}(\hat\bx, x_k)\big),\quad j = 1, 2, \dots, N-1.
\end{align*}
Representing 
\begin{align*}
\b_{j,k}(\hat\bx, x) - 
\b_{j,k}(\hat\bx, x_k)
= (x-x_k)\cdot \int_0^1 \nabla_x \b_{j,k}(\hat\bx, x_k+ s(x-x_k)) ds 
=: (x_k-x) \cdot\s_{j,k}(\hat\bx, x), 
\end{align*}
we can rewrite $\CF_{j,k}$ as 
\begin{align*}
\CF_{j,k}(\hat\bx, x) = 
\Xi_{j,k}(\hat\bx, x) 
b_{j,k}(\hat\bx) a(x), 
\quad  \textup{where}
\quad \Xi_{j,k}(\hat\bx, x) = \Phi(x_k-x)\,\big[(x_k-x)\cdot\s_{j,k}(\hat\bx, x)\big].
\end{align*}
Remembering that $\Phi$ is homogeneous of order $\a$  
and that $\s_{j,k}\in\plainC\infty_0(\R^{3N})$, 
we conclude that 
\begin{align*}
\big|\p^m_x \Xi_{j,k}(\hat\bx, x)\big|
\lesssim |x-x_j|^{\a+1-|m|}, \quad m\in \mathbb N_0^3.
\end{align*}
Since $\s_{j,k}$ is compactly supported, 
the kernel $\Xi_{j,k}(\hat\bx, x)$, as a function of $x\in \R^3$, 
belongs to $\plainH{l}(\R^3)$ for all $0\le l<\a+5/2$. 
As $\a > -5/2$ the set of such values $l$ is non-empty. 
Moreover, the $\plainH{l}$-norm of the kernel, as a function of $\hat\bx\in\R^{3N-3}$,  
is uniformly bounded, and hence 
it trivially belongs to $\plainL2(\R^{3N-3}, |b_{j,k}(\hat\bx)|^2 d\hat\bx)$.
By virtue of Proposition \ref{prop:BS}, we obtain 
 that $\SfG_{q}(\iop(\CF_{j,k})) = 0$, $1/q = 1/2+ l/3$.  Note that 
\begin{align*}
\frac{1}{q} \ge \frac1{p} =  1+\frac{\a}{3},\quad \textup{for}\quad l\ge \a+\frac{3}{2}. 
\end{align*}
Consequently, taking $l$ to be the only non-negative integer in the interval $[\a+3/2, \a+5/2)$, 
we conclude that $\SfG_p(\iop(\CF_{j,k})) = 0$, 
and, by \eqref{eq:triangleg}, 
\begin{align*}
\SfG_p \bigg(\sum_{k=1}^{N-1}\iop(\CF_{j,k})\bigg) = 0.
\end{align*}
By Corollary \ref{cor:zero},
\begin{align}\label{eq:btoxi}
\SfG_p\big(\iop(\CM)\big)  
= \SfG_p\big(\iop(\CA)\big),\quad \sg_p\big(\iop(\CM)\big)  
= \sg_p\big(\iop(\CA)\big),
\end{align}
where $\CA$ is the vector function with components $\CA_j(\hat\bx, x)$, 
$j = 1, 2, \dots, N$. 
To find $\SfG_p$ and $\sg_p$
for the operator $\iop(\CA)$, we observe that each kernel $\CA_j(\hat\bx, x)$ has the form 
\begin{align*}
\sum_{k=1}^{N-1} \Phi(x_k-x) \tilde b_{j,k}(\hat\bx) a(x) \quad \textup{with}
\quad \tilde b_{j,k}(\hat\bx) = b_{j,k}(\hat\bx) \b_{j,k}(\hat\bx, x_k),
\end{align*}
Using the result of Subsect \ref{subsect:beta1} we obtain the formula 
\eqref{eq:model} for the operator $\iop(\CA)$ 
 with the function $h$ defined in \eqref{eq:hm}. 
In view of \eqref{eq:btoxi} this implies 
 \eqref{eq:model} for the operator $\iop(\CM)$, as claimed. 
\qed

\begin{cor}\label{cor:scal1} 
Let $\Phi(x)$ be as in Example \ref{ex:scal} with $\a = 1, d = 3$, i.e. 
$\Phi(x) = |x|$ and $1/p = 1+\a/d = 4/3$.
According to \eqref{eq:model} and \eqref{eq:Xint},  
 \begin{align}\label{eq:scal1} 
\SfG_{3/4}\big( 
\iop(\CM)\big) 
= \sg_{3/4}\big( 
\iop(\CM)\big)
= \mu_{1, 3}\int_{\R^3} \big(|a(x) h(x)| \big)^{\frac{3}{4}}\,  dx,
\end{align}
where $\mu_{\a, d}$ is defined in \eqref{eq:Xint}.
\end{cor}

\begin{cor}\label{cor:grad1} 
Let  $\Phi(x)$ be as in Example \ref{ex:grad} with $\a = 0, d = 3$, i.e. 
$\Phi(x) = \nabla|x| = |x|^{-1} x$ and $1/p = 1+\a/d = 1$.
According to \eqref{eq:model}  
and \eqref{eq:Xintvec},    
 \begin{align*}
\SfG_1\big(\iop(\CM)\big)
= \sg_1\big(\iop(\CM)\big)
= \nu_{0, 3}\int_{\R^3} |a(x) h(x)|\, dx,
\end{align*}
where $\nu_{\a, d}$ is defined in \eqref{eq:Xintvec}.
\end{cor}

In the current paper we need only Corollary \ref{cor:scal1}. 
Corollary \ref{cor:grad1} is needed for future use.   

\section{Factorization of $\BOLG$: operator $\BPsi$}\label{sect:factor}

\subsection{Reformulation of the problem} 
Using the functionals \eqref{eq:limsupinf}, one can rewrite the sought formula \eqref{eq:main} as 
\begin{align*}
\SfG_{3/8}(\BOLG) = \sg_{3/8}(\BOLG) = A. 
\end{align*}
Since $\BOLG = \BPsi^*\BPsi$ with the operator $\BPsi: \plainL2(\R^3)\to\plainL2(\R^{3N-3})$ 
defined in \eqref{eq:bpsi}, 
by \eqref{eq:double} the above equalities rewrite as
\begin{align}\label{eq:gtopsi}
\SfG_{3/4}(\BPsi) = \sg_{3/4}(\BPsi) = A.
\end{align}
Thus the main Theorem \ref{thm:maincompl} can be recast as follows:

\begin{thm}\label{thm:gtopsi}
Under the conditions of Theorem \ref{thm:maincompl} the 
formula \eqref{eq:gtopsi} holds with the constant $A$ 
which is defined in \eqref{eq:coeffA}.
\end{thm}

The rest of the paper is focused on the proof of Theorem \ref{thm:gtopsi}.

As explained in the Introduction, 
at the heart of the proof is the formula \eqref{eq:repr} for the function $\psi$, 
which translates to the representation \eqref{eq:trepr} for the kernels $\psi_j$ defined in 
\eqref{eq:psij}.  
This representation allows us to reduce the problem to the model operator 
considered in Sect. \ref{sect:model} with the function $\Phi(x) = |x|$. 
At the first stage of this reduction we construct $\plainC\infty_0$ approximations of the 
functions $\tilde\xi_{j,k}$ and $\tilde\eta_{j, k}$ from \eqref{eq:trepr}.  

\subsection{Cut-off functions} 
Firt we construct appropriate cut-offs. 
Fix a $\d > 0$. Along with the sets \eqref{eq:sls} introduce 
\begin{align}\label{eq:slsd}
\SfS_{l, s}(\d) = \SfS_{s, l}(\d) = \{\bx\in\R^{3N}: |x_l-x_s|>\d\},\ 0\le l < s\le N,
\end{align} 
and for all $k = 0, 1, \dots, N-1$, define
\begin{align}\label{eq:ujd}
\SU_{k}(\d) = \bigg(\bigcap_{0\le l < s\le N-1} \SfS_{l, s}(\d)\bigg)\bigcap \bigg(
\bigcap_{\substack{0\le s\le N-1\\s\not = k}} \SfS_{s, N}(\d)\bigg).
\end{align}
Comparing with \eqref{eq:uj} we see that $\SU_k(\d)\subset \SU_{k, N}$, 
and for $\bx\in \SU_k(\d)$ all 
the coordinate pairs, except for $x_k$ and $x_N$, are separated by a distance $\d$.
Similarly to \eqref{eq:diag} define the diagonal set 
\begin{align*}
\SU^{(\rm d)}_{k}(\d) = \{\bx\in\SU_{k}(\d): x_j = x_N\}\subset\SU_{k, N}^{(\dc)}.
\end{align*}
Recall that the representation 
\eqref{eq:trepr} holds on the domain $\tilde\Om_{j, k}$ which satisfies 
\eqref{eq:tom} for all $j = 1, 2, \dots, N$, $k = 0, 1, \dots, N-1$. 
We construct a compact subset of $\tilde\Om_{j, k}$ in the following way. 
For $R>0$ let 
\begin{align*}
\SU_{k}(\d, R) = &\  \SU_k(\d)\bigcap \ (B_R)^N,\\ 
\SU^{(\rm d)}_{k}(\d, R) = &\ \{\bx\in\SU_{k}(\d, R): x_k = x_N\},
\end{align*}
where $B_R = \{x\in\R^3: |x| <R\}$. 
The  set $\SU^{(\rm d)}_{k}(\d, R)$ is bounded and its closure belongs to 
$\tilde\Om_{j,k}$ for all $\d>0, R>0$. Therefore, there exists 
an $\varepsilon_0 = \varepsilon_0(\d, R)>0$ such that the $\varepsilon$-neighbourhood 
\begin{align}\label{eq:omj}
\tilde\Om_{k}(\d, R, \varepsilon) := \{\bx\in\SU_{k}(\d, R): |x_k-x_N|<\varepsilon\},
\end{align}
together with its closure, belongs to $\tilde\Om_{j,k}$ 
for all $\varepsilon\in (0, \varepsilon_0)$:
\begin{align}\label{eq:omdere}
\overline{\tilde\Om_{k}(\d, R, \varepsilon)}\subset \tilde\Om_{j,k},\quad \forall 
\varepsilon\in (0, \varepsilon_0).
\end{align} 

Now we specify $\plainC\infty_0$ cutoffs supported on the domains 
$\tilde\Om_{k}(\d, R, \varepsilon)$. 
Let $\t\in\plainC\infty_0(\R)$ and $\z = 1-\t$ be as defined in \eqref{eq:sco}, \eqref{eq:sco1}. 
Denote 
\begin{align}\label{eq:ydel}
Y_\d(\hat\bx) = \prod_{0\le l < s \le N-1} \z\big(|x_l-x_s|(4\d)^{-1}\big).
\end{align}
By the definition of $\z$, 
\begin{align}\label{eq:omdel}
\supp Y_\d\subset 
\bigcap_{0\le l<s\le  N-1} \SfS_{l, s}(2\d),
\end{align}
where $\SfS_{l, s}(\ \cdot\ )$ is defined in \eqref{eq:slsd}. 
Define also cut-offs at infinity. Denote
\begin{align}\label{eq:qr} 
Q_R(\hat\bx) = \prod_{1\le l\le N-1} \t\big(|x_l|R^{-1}\big),\quad 
K_R(x) = \t\big(|x|R^{-1}\big).
\end{align}

\begin{lem}\label{lem:cutoff}
Let $\tilde\Om_k(\d, R, \varepsilon)$ be the set introduced in \eqref{eq:omj}. 
Then for all $\varepsilon<\min\{\varepsilon_0, \d\}$ the support of the function 
\begin{align}\label{eq:cutoff}
Q_R(\hat\bx) K_R(x) Y_\d(\hat\bx) \t\big(|x-x_k|\varepsilon^{-1}\big)
\end{align}
belongs to $\tilde\Om_k(\d, R, \varepsilon)$ for all $k = 0, 1, \dots, N-1$. 
\end{lem}

\begin{proof} 
Assume that $\bx$ belongs to the support of the function \eqref{eq:cutoff}. 
In view of \eqref{eq:omdel}, for such $\bx$  
we have 
\begin{align}\label{eq:omdel1}
|x_l-x_s|>2\d, \ 0\le l< s\le N-1,\quad \textup{and}\quad |x-x_k|<\varepsilon.
\end{align}
As $\varepsilon<\d$, for all $s = 0, 1, \dots, N-1$, $s\not = k$, we can write
\begin{align*}
|x-x_s|\ge |x_k-x_s| - |x-x_k|>  2\d - \varepsilon > \d,
\end{align*}
By definition \eqref{eq:ujd}, together with 
\eqref{eq:omdel1} this gives 
$\bx\in \SU_k(\d)$. Moreover, since $\supp (Q_R K_R)\subset (B_R)^N$, this means that 
$\bx\in \SU_k(\d, R)$.  
Now the claimed inclusion follows from the definition \eqref{eq:omj}. 
\end{proof}

\subsection{} 
Using the cut-offs introduced above we construct a convenient approximation 
for the kernels $\psi_j(\hat\bx, x)$. 
Taking if necessary, a smaller $\varepsilon_0$ in \eqref{eq:omdere},  
we will assume that $\varepsilon_0(\d, R) \le \d$, and hence for all 
$\varepsilon < \varepsilon_0(\d, R)$, apart from the inclusion \eqref{eq:omdere} 
we have Lemma \ref{lem:cutoff}. 
Thus, for these values of $\varepsilon$ 
the real analytic functions $\tilde\xi_{j,k}, \tilde\eta_{j,k}$ 
are well-defined on the support of \eqref{eq:cutoff}, 
and hence the kernel 
\begin{align}\label{eq:upsilon}
\Upsilon_j[\d,R,\varepsilon](\hat\bx, x) 
=  Q_R(\hat\bx) Y_\d(\hat\bx) K_R(x)\sum_{k=1}^{N-1}\t\big(|x-x_k|\varepsilon^{-1}\big)
|x-x_k|\tilde\eta_{j, k}(\hat\bx, x),
\end{align}
is well-defined for all $(\hat\bx, x)\in\R^{3N}$, and each of the functions
\begin{align*}
Q_R(\hat\bx) Y_\d(\hat\bx) K_R(x) \t\big(|x-x_k|\varepsilon^{-1}\big)
 \tilde\eta_{j,k}(\hat\bx, x),\quad k = 1, 2, \dots, N-1,
\end{align*}
is $\plainC\infty_0(\R^{3N})$. 
Our objective is to prove that the vector-valued kernel 
\begin{align*}
\boldsymbol\Upsilon[\d, R, \varepsilon](\hat\bx, x) 
= \big\{\Upsilon_j[\d,R,\varepsilon](\hat\bx, x)\big\}_{j=1}^N
\end{align*} 
is an approximation for $\BPsi(\hat\bx, x)$(see \eqref{eq:bpsi}) in the following sense. 

\begin{lem}\label{lem:central} 
The following relations hold: 
\begin{align*}
\SfG_{3/4}(\BPsi) = \lim\limits_{\substack{\d\to 0\\ R\to\infty}} \lim_{\varepsilon\to 0}
\SfG_{3/4}\big(\iop(\boldsymbol\Upsilon[\d, R, \varepsilon])\big),\quad 
\sg_{3/4}(\BPsi) = \lim\limits_{\substack{\d\to 0\\ R\to\infty}} \lim_{\varepsilon\to 0}
\sg_{3/4}\big(\iop(\boldsymbol\Upsilon[\d, R, \varepsilon])\big), 
\end{align*}
where the limits on the right-hand side exist. 
\end{lem}

The proof of this lemma is given in the next section.

\section{Proof of Lemma \ref{lem:central}}\label{sect:trim}

\subsection{Spectral estimates for $\BPsi$} 
Our proof of Lemma \ref{lem:central} relies on the bounds 
obtained in \cite{Sobolev2020}. 
Let $\CC_n = (0, 1)^3 + n$,\ $n\in\mathbb Z^3$. 
Assume that $b\in\plainL\infty(\R^{3N-3})$ and that 
$a\in\plainL2_{\textup{\tiny loc}}(\R^3)$ is such that  
\begin{align*}
\sup_{n\in\mathbb Z^3} \|a\|_{\plainL2(\CC_n)}<\infty.
\end{align*}
Then the functionals 
\begin{align*}
S_\varkappa(a) = \bigg[\sum_{n\in\mathbb Z^3} e^{- \frac{3}{4}\varkappa|n|}
\|a\|_{\plainL2(\CC_n)}^{\frac{3}{4}}\bigg]^{\frac{4}{3}} 
\end{align*}
and 
\begin{align*}
M_\varkappa(b) = \biggl[\int_{\R^{3N-3}} 
|b(\hat\bx)|^2 e^{-2\varkappa|\hat\bx|} 
d\hat\bx\biggr]^{\frac{1}{2}},
\end{align*}
are both finite for all $\varkappa >0$. 
Recall that the functional $\SfG_p$ is defined in \eqref{eq:limsupinf}, 
and $\psi_j$ -- in \eqref{eq:psij}.  
The next bound for the operators 
$b\,\iop(\psi_j)a$ 
follows from \cite[Theorem 3.1]{Sobolev2020}. 

\begin{prop} 
Assume that $\psi$ satisfies \eqref{eq:exp}, and let $j = 1, 2, \dots, N$.  
Let the functions $a$ and $b$ be as described above. 
Then  $b\, \iop(\psi_j) a\in \BS_{3/4, \infty}$ and for some $\varkappa \le \varkappa_0$ we have  
\begin{align}\label{eq:psifull}
 \SfG_{3/4}(b\,\iop(\psi_j) a)
\lesssim \big(M_\varkappa(b) S_\varkappa(a)\big)^{\frac{3}{4}}.
\end{align}
\end{prop}

\subsection{Proof of Lemma \ref{lem:central}}
 The strategy of the proof is to ``trim down" the kernel \eqref{eq:bpsi} in several steps, 
 by multiplying it by appropriate cut-offs including the functions \eqref{eq:ydel} and \eqref{eq:qr}, or dropping some of the components, 
 until it reduces to the kernel \eqref{eq:upsilon}. 
 At every step of this process we justify the trimming using either Corollary \ref{cor:zero} 
 or Corollary \ref{cor:zero1}.

The first stage is described in the next lemma.

\begin{lem}
The following relations hold: 
\begin{align}\label{eq:asymp1}
\SfG_{3/4}(\BPsi) = \lim\limits_{\substack{\d\to 0\\ R\to\infty}} \SfG_{3/4}(Q_RY_\d \BPsi K_R),\quad 
\sg_{3/4}(\BPsi) = \lim\limits_{\substack{\d\to 0\\ R\to\infty}} \sg_{3/4}(Q_RY_\d \BPsi K_R),
\end{align}
where the limits on the right-hand side exist. 
\end{lem}

\begin{proof} First we check that 
\begin{align}\label{eq:psidelR}
\begin{cases}
\lim\limits_{\d\to 0} \SfG_{3/4}\big((I- Y_\d)\BPsi\big) = 0,\\[0.3cm] 
\lim\limits_{R\to \infty} \SfG_{3/4}\big((I- Q_R)\BPsi\big) = 0,\
\lim\limits_{R\to \infty} \SfG_{3/4}\big(\BPsi(I- K_R)\big) = 0. 
\end{cases}
\end{align}
It suffices to check the above relations for each operator $\iop(\psi_j)$, $j = 1, 2, \dots, N$. 
Consider first $(I-Y_\d)\iop(\psi_j)$. 
Since 
\begin{align*}
1-Y_\d(\hat\bx)\le \sum_{0\le l<s\le N-1} \t\big(|x_l-x_s|(4\d)^{-1}\big),
\end{align*}
it follows from \eqref{eq:psifull} that 
\begin{align*}
\SfG_{3/4}\big((1- Y_\d)\iop(\psi_j)\big)
\lesssim &\ \big(M_\varkappa(1-Y_\d)\big)^{3/4}\\
\lesssim &\ \sum_{0\le l < s\le N-1} 
\bigg[
\int \t\big(|x_l-x_s|(4\d)^{-1}\big)^2 e^{-2\varkappa |\hat\bx| } d\hat\bx
\bigg]^{3/8}\lesssim \d^{9/8}\to 0,\ \d\to 0,
\end{align*}
and hence the first relation in \eqref{eq:psidelR} holds. 

In a similar way one estimates $(I-Q_R)\iop(\psi_j)$ and 
$\iop(\psi_j)(I-K_R)$. Estimate, for example, the first of these operators. 
Since 
\begin{align*}
1-Q_R(\hat\bx)\le \sum_{1\le l\le N-1} \z\big(|x_l|R^{-1}\big),
\end{align*}
it follows from \eqref{eq:psifull} again that 
\begin{align*}
\SfG_{3/4}\big((I-Q_R)\iop(\psi_j)\big)\lesssim &\ \big(M_\varkappa(1-Q_R)\big)^{3/4}\\
\lesssim &\ \sum_{0\le l\le N-1}
\bigg[
\int_{\R^{3N-3}} \z(|x_l| R^{-1})^2  e^{-2\varkappa |\hat\bx| }\, d\hat\bx
\bigg]^{3/8}\lesssim  e^{-3\varkappa R/8}\to 0, \ R\to\infty, 
\end{align*}
whence the second equality in \eqref{eq:psidelR}.  

Represent $\BPsi$ in the form 
\begin{align*}
\BPsi = Q_R Y_d\BPsi K_R + (I-Q_R)\BPsi + Q_R(1-Y_\d)\BPsi + 
Q_R Y_\d \BPsi(I-K_R),
\end{align*}
According to \eqref{eq:triangleg}, 
\begin{align*}
\SfG_{3/4}\big(\BPsi - Q_R Y_d\BPsi K_R\big)^{\frac{\scalel{3}}{\scalel{7}}}
\le &\ \SfG_{3/4}\big((I-Q_R)\BPsi\big)^{\frac{\scalel{3}}{\scalel{7}}}\\[0.2cm] 
&\ + \SfG_{3/4}\big(Q_R(1-Y_\d)\BPsi\big)^{\frac{\scalel{3}}{\scalel{7}}} + 
\SfG_{3/4}\big(Q_R Y_\d \BPsi(I-K_R)\big)^{\frac{\scalel{3}}{\scalel{7}}}\\[0.2cm]
\le &\ \SfG_{3/4}\big((I-Q_R)\BPsi\big)^{\frac{\scalel{3}}{\scalel{7}}}\\[0.2cm] 
&\ + \SfG_{3/4}\big((1-Y_\d)\BPsi\big)^{\frac{\scalel{3}}{\scalel{7}}} + 
\SfG_{3/4}\big(\BPsi(I-K_R)\big)^{\frac{\scalel{3}}{\scalel{7}}}.
\end{align*}
By virtue of \eqref{eq:psidelR} the right-hand side tends to 
zero as $\d\to 0, R\to\infty$.
By Corollary \ref{cor:zero1} this implies \eqref{eq:asymp1}. 
\end{proof}

At the next stage we partition the kernel 
\begin{align}\label{eq:trim}
Q_R(\hat\bx) Y_\d(\hat\bx) \BPsi(\hat\bx, x) K_R(x)
\end{align}
of the operator $Q_R Y_\d\BPsi K_R$ on the right-hand side of the formulas 
\eqref{eq:asymp1}. 
We do this by introducing the cut-offs $\t\big(|x-~x_k|\varepsilon^{-1}\big)$, 
$k = 0, 1, \dots, N-1$, assuming that $\varepsilon<\d$. 
In view of the definition 
\eqref{eq:ydel} it is straightforward to check that under this condition, we have 
\begin{align*}
Y_\d(\hat\bx)\sum_{k=0}^{N-1} \t\big(|x-x_k|\varepsilon^{-1}\big) 
+ Y_\d(\hat\bx)\prod_{k=0}^{N-1} \z\big(|x-x_k|\varepsilon^{-1}\big) = Y_\d(\hat\bx),
\end{align*} 
and hence the $j$'th component of \eqref{eq:trim} can be represented as follows:
\begin{align}\label{eq:split}
Q_R(\hat\bx) Y_\d(\hat\bx) \psi_j(\hat\bx, x) K_R(x)
= \sum_{k=0}^{N-1}\phi_{j,k}[\d, R, \varepsilon](\hat\bx, x) + \tau_j[\d, R, \varepsilon](\hat\bx, x)
\end{align}
with 
\begin{align*}
\phi_{j,k}[\d, R, \varepsilon](\hat\bx, x)
= &\ Q_R(\hat\bx) Y_\d(\hat\bx)\t\big(|x-~x_k|\varepsilon^{-1}\big) \psi_j(\hat\bx, x) K_R(x),\quad 
k = 0, 1, \dots, N-1,\\[0.2cm]
\tau_j[\d, R, \varepsilon](\hat\bx, x)
= & \  Q_R(\hat\bx) 
Y_\d(\hat\bx) \prod_{k=0}^{N-1} \z\big(|x-x_k|\varepsilon^{-1}\big)  \psi_j(\hat\bx, x) K_R(x).
\end{align*} 
First we show that the kernels $\tau_j[\d, R, \varepsilon]$ and $\phi_{j,0}[\d, R, \varepsilon]$ 
give negligible contributions to the asymptotics. 

\begin{lem} For each $\d>0, R>0$ and $\varepsilon<\d$ one has 
\begin{align}\label{eq:tau}
\SfG_{3/4}\big(\iop\big(\tau_j[\d, R, \varepsilon]\big)\big) = 0,\ j = 1, 2, \dots, N.
\end{align}
\end{lem}

\begin{proof} 
By the definitions \eqref{eq:ydel} and \eqref{eq:sco1}, 
the support of the kernel 
$\tau_j[\d, R, \varepsilon]$ belongs to the bounded domain 
\begin{align*}
 {\bigcap_{0\le l < s \le N}  \SfS_{l, s}(\varepsilon/2)\cap (B_R)^N}. 
\end{align*}
The function $\psi_j$ is real-analytic on this domain and it is uniformly bounded 
together with all its derivatives, so that $\tau_j[\d, R, \varepsilon]\in \plainC\infty_0(\R^{3N})$.
By Proposition \ref{prop:BS}, 
$\SfG_p(\iop(\tau_j[\d, \R, \varepsilon])) = 0$ for all $p >0$, and in particular, 
for $p = 3/4$, as claimed. 
\end{proof}
 
\begin{lem} For each $\d>0, R>0$ one has 
\begin{align}\label{eq:phi0}
\lim_{\varepsilon\to 0}\SfG_{3/4}\big(\iop\big(\phi_{j,0}[\d, R, \varepsilon]\big)\big) = 0,\ 
j = 1, 2, \dots, N.
\end{align}
\end{lem}

\begin{proof}
As $x_0 = 0$ by definition, the kernel $\phi_{j,0}[\d, R, \varepsilon]$ has the form 
\begin{align*}
\phi_{j,0}[\d, R, \varepsilon](\hat\bx, x)
= Q_R(\hat\bx) Y_\d(\hat\bx)\psi_j(\hat\bx, x) \t\big(|x|\varepsilon^{-1}\big)  K_R(x). 
\end{align*}
Estimating $Q_R Y_\d\le 1$, $ K_R\le 1$, one sees that 
the singular values of $\iop(\phi_{j,0}[\d, R, \varepsilon])$ do not exceed those of the operator 
$\iop(\psi_j) a$ with the weight $a(x) = \t(|x|\varepsilon^{-1})$. 
By \eqref{eq:psifull}, 
\begin{align*}
\SfG_{3/4}(\iop(\psi_j)a)\lesssim S_\varkappa(a)^{3/4}\lesssim \bigg(\int_{\R^3} \t\big(|x|\varepsilon^{-1} \big)^2 dx\bigg)^{3/8}
\lesssim \varepsilon^{9/8}\to 0,\ \varepsilon\to 0.
\end{align*}
This implies \eqref{eq:phi0}.
\end{proof}

\begin{cor}
Denote by $\boldsymbol\a[\d, R, \varepsilon](\bx,x) = \{\a_j[\d, R, \varepsilon]\}_{j=1}^N$ 
the vector-valued kernel with the components  
\begin{align*}
\a_j[\d, R, \varepsilon](\hat\bx, x) 
= \sum_{k=1}^{N-1}\phi_{j,k}[\d, R, \varepsilon](\hat\bx, x).
\end{align*}
Then for all $\d >0$ and $R>0$, we have  
\begin{align}\label{eq:asymp2}
\begin{cases}
\SfG_{3/4}(Q_RY_\d \BPsi K_R) = &\ 
\lim\limits_{\varepsilon\to 0} 
\SfG_{3/4}(\iop(\boldsymbol\a[\d, R, \varepsilon])),
\\[0.2cm]
\sg_{3/4}(Q_RY_\d \BPsi K_R) = &\ 
\lim\limits_{\varepsilon\to 0} 
\sg_{3/4}(\iop(\boldsymbol\a[\d, R, \varepsilon])),
\end{cases}
\end{align}
where the limits on the right-hand side exist.
\end{cor}

\begin{proof}
By \eqref{eq:split}, the kernel $Q_R Y_\d \psi_j K_R$ has the form 
\begin{align*}
\a_j[\d, R, \varepsilon] + \phi_{j, 0}[\d, R, \varepsilon] + \tau_j[\d, R, \varepsilon].
\end{align*}
By virtue of \eqref{eq:triangleg} and \eqref{eq:tau}, \eqref{eq:phi0}, we have 
\begin{align*}
\lim_{\varepsilon\to 0}
\SfG_{3/4}(\iop\big(\phi_{j, 0}[\d, R, \varepsilon] + \tau_j[\d, R, \varepsilon]\big) = 0.
\end{align*}
Now \eqref{eq:asymp2} follows from Corollary \ref{cor:zero1}.
\end{proof}

\begin{proof}[Completion of the proof of Lemma \ref{lem:central}]
According to Lemma \ref{lem:cutoff}, under the condition $\varepsilon < \varepsilon_0(\d, R)$, the 
support of each kernel
\begin{align*}
\phi_{j, k}[\d, R, \varepsilon], \quad j= 1, 2, \dots, N,\quad k = 1, 2, \dots, N-1,
\end{align*}  
belongs to $\tilde\Om_k(\d, R, \varepsilon)$, see \eqref{eq:omj} 
for the definition. Therefore one can use the representation \eqref{eq:trepr} for the function 
$\psi_j$: 
\begin{align*}
\a_j[\d, R, \varepsilon](\hat\bx, x) 
= \sum_{k=1}^{N-1}&\ \phi_{j,k}[\d, R, \varepsilon](\hat\bx, x)
= \sum_{k=1}^{N-1} Q_R(\hat\bx) Y_\d(\hat\bx)\t\big(|x-x_k|\varepsilon^{-1}\big) 
\tilde\xi_{j,k}(\hat\bx, x) K_R(x)\\[0.2cm]
&\ \quad + \sum_{k=1}^{N-1} Q_R(\hat\bx) Y_\d(\hat\bx)\t\big(|x-x_k|\varepsilon^{-1}\big) 
|x_k-x|\tilde\eta_{j,k}(\hat\bx, x) K_R(x).
\end{align*}
Each term in the first sum on the right-hand side is $\plainC\infty_0(\R^{3N})$.
 Thus, by Proposition \ref{prop:BS}, the functional 
$\SfG_p$ for the associated operator equals zero for all $p >0$, and in particular, for $p=3/4$.
The second sum coincides with the kernel $\Upsilon_j[\d, R, \varepsilon](\hat\bx, x)$, 
defined in \eqref{eq:upsilon}. Therefore, by Corollary \ref{cor:zero}, 
\begin{align}\label{eq:asymp3} 
\begin{cases}
\SfG_{3/4}(\iop(\boldsymbol\a[\d, R, \varepsilon])) = 
\SfG_{3/4}(\iop(\boldsymbol\Upsilon[\d, R, \varepsilon])),\\[0.2cm]
\sg_{3/4}(\iop(\boldsymbol\a[\d, R, \varepsilon])) = 
\sg_{3/4}(\iop(\boldsymbol\Upsilon[\d, R, \varepsilon])),
\end{cases}
\end{align}
for each $\d>0, R>0$ and $\varepsilon<\varepsilon_0(\d, R)$. 
Putting together \eqref{eq:asymp1}, \eqref{eq:asymp2} and \eqref{eq:asymp3}, 
and using Corollary \ref{cor:zero1}, 
we conclude the proof of Lemma \ref{lem:central}.
\end{proof}

\section{Proof of Theorems \ref{thm:etadiag} and \ref{thm:gtopsi}, \ref{thm:maincompl}}
\label{sect:proofs}

\begin{lem} 
The operator $\iop(\boldsymbol\Upsilon[\d, R, \varepsilon])$ belongs to $\BS_{3/4, \infty}$ 
for all $\d >0, R>0, \varepsilon<\varepsilon_0(\d, R)$ 
 and 
\begin{align}\label{eq:upsas}
\SfG_{3/4}\big(\iop(\boldsymbol\Upsilon[\d, R, \varepsilon])\big) 
= \sg_{3/4}\big(\iop(\boldsymbol\Upsilon[\d, R, \varepsilon])\big) 
= \mu_{1, 3} \int \big(K_R(t) H_{\d, R}(t) \big)^{\frac{3}{4}}dt,
\end{align}
where
\begin{align*}
H_{\d, R}(t) = Q_R(t) Y_\d(t)\, \big(|\tilde\eta_{1, 1}(t, t)|^2 + \tilde\eta_{1, 2}(t, t)|^2\big)^{1/2},\ \textup{if}\ N=2,
\end{align*}
and
\begin{align}\label{eq:hdr}
H_{\d, R}(t) =  
\bigg[\sum_{j=1}^N\sum_{k=1}^{N-1}\int_{\R^{3N-6}} \big| 
Q_R(\tilde\bx_{k, N}, t) Y_\d(\tilde\bx_{k, N}, t)
\tilde\eta_{j, k}(\tilde\bx_{k, N}, t, t) \big|^2 d\tilde\bx_{k, N}\bigg]^{\frac{1}{2}},\ \textup{if}\ N\ge 3,
\end{align} 
and $\mu_{\a, d}$ is defined in \eqref{eq:Xint}.
\end{lem} 
 
\begin{proof}
The kernel $\boldsymbol\Upsilon[\d, R, \varepsilon]$ (see \eqref{eq:upsilon}) has the form \eqref{eq:cm} 
with 
\begin{align*}
a(x) = K_{2R}(x),\ &\ b_{j,k}(\hat\bx) = Q_{2R}(\hat\bx) Y_{\d/2}(\hat\bx),\\ 
\b_{j, k}(\hat\bx, x) = &\ \t\big(|x-x_k|\varepsilon^{-1}\big) \tilde\eta_{j, k}(\hat\bx, x)
Q_{R}(\hat\bx) Y_{\d}(\hat\bx)K_{R}(x),
\end{align*}
and the homogeneous function $\Phi(x) = |x|$.   
Here we have used the fact that 
\begin{align*}
Q_R(\hat\bx)Q_{2R}(\hat\bx) = Q_R(\hat\bx),\quad 
Y_\d(\hat\bx) Y_{\d/2}(\hat\bx) = Y_\d(\hat\bx)
\quad \textup{and} \quad K_R(x) K_{2R}(x) = K_R(x).
\end{align*}
Therefore we can use Corollary \ref{cor:scal1}. It is immediate to see that in this case 
the function $h$ defined in \eqref{eq:hm}, coincides with $H_{\d, R}$, so that 
\eqref{eq:scal1} entails \eqref{eq:upsas}, as required. 
\end{proof}

\begin{proof}[Proof of Theorems \ref{thm:etadiag}, \ref{thm:gtopsi} and \ref{thm:maincompl}]
By Lemma \ref{lem:central}, each term in the relation \eqref{eq:upsas} has a limit 
as $\d\to 0, R\to\infty$. Therefore the integral on the right-hand side of 
\eqref{eq:upsas} is bounded uniformly in $\d>0, R>0$. 
Assume for convenience that the function $\t$ defined 
in \eqref{eq:sco1} is monotone decreasing for $t\ge 0$. Therefore the pointwise 
convergencies
\begin{align*}
Y_\d(\tilde\bx_{k, N}, t)\to 1, \ \d\to 0\quad \textup{and}\quad 
K_R(t)\to 1, Q_R(\tilde\bx_{k, N} , t)\to 1,\ R\to\infty,
\end{align*} 
are monotone increasing. By the Monotone Convergence Theorem, 
the integrand $ K_R(t) H_{\d, R}(t) $ on the right-hand side of \eqref{eq:upsas} 
converges for a.e. $t\in\R^3$ as $\d\to 0, R\to\infty$ to an $\plainL{3/4}(\R)$-function,  
which we denote by $\tilde H(t)$, and the 
integral in \eqref{eq:upsas} converges to 
\begin{align}\label{eq:limit}
\mu_{1, 3}\int \big(\tilde H(t)\big)^{3/4} dt.  
\end{align}
If $N=2$, then this concludes the proof 
of Theorem 
\ref{thm:etadiag}, 
since in this case 
\begin{align*}
H_{\d, R}(t)\to\big(|\tilde\eta_{1, 1}(t, t)|^2 + \tilde\eta_{2, 1}(t, t)|^2\big)^{1/2},
\end{align*} 
a.e. $t\in\R^3$, and by virtue of \eqref{eq:tH} this limit coincides with $H(t)$. 

If $N\ge 3$, then the convergence to $\tilde H(t)$ implies that  
for a.e. $t\in \R^3$ 
the function $K_R(t)H_{\d, R}(t)$, and hence $H_{\d, R}(t)$, is bounded uniformly in $\d$ and $R$. 
Applying the Monotone Convergence Theorem to the integral \eqref{eq:hdr}, 
we conclude that the a.e.-limit
\begin{align*}
|\tilde\eta_{j,k}(\tilde\bx_{k, N}, t, t)| 
= \lim_{\d\to 0, R\to\infty}\big| Q_R(\tilde\bx_{k, N}, t) Y_\d(\tilde\bx_{k, N}, t)
\tilde\eta_{j,k}(\tilde\bx_{k, N}, t, t) \big|,\ 
\end{align*}
belongs to $\plainL2(\R^{3N-6})$,\ a.e. $t\in\R^3$, and 
\begin{align*}
\lim_{\d\to 0, R\to\infty} H_{\d, R}(t) = H(t),\quad \textup{a.e.}\quad t\in \R^3,
\end{align*}
 where we have used the formula \eqref{eq:tH} for $H$. Thus 
 $H = \tilde H\in\plainL{3/4}(\R^3)$. As 
\eqref{eq:tH} is equivalent to \eqref{eq:H}, 
this completes the proof of Theorem \ref{thm:etadiag}. 
  
 An easy calculation shows that 
 $\mu_{1, 3} = 3^{-1}(2/\pi)^{5/4}$, so that the limit \eqref{eq:limit} 
 coincides with the coefficient 
 $A$ in \eqref{eq:coeffA}. Together with Lemma \ref{lem:central} this completes the proof of Theorem 
 \ref{thm:gtopsi}.
 As explained before,  Theorem \ref{thm:gtopsi} is 
 equivalent to Theorem \ref{thm:maincompl}. This completes the proof.
\end{proof}  
    
  \textbf{Acknowledgments.} The author is grateful to S. Fournais, T. Hoffmann-Ostenhof, 
M. Lewin and T. \O. S\o rensen
for stimulating discussions and advice. 
The author thanks J. Cioslowski for his comments and 
for bringing to the author's attention  
papers \cite{Cioslowski2020}, 
\cite{CioPrat2019} and \cite{HaKlKoTe2012}. 

The author was supported by the EPSRC grant EP/P024793/1.

\bibliographystyle{../beststyle}

\end{document}